\documentclass[a4paper,onecolumn,11pt,accepted=2024-10-15]{quantumarticle}
\pdfoutput=1
\usepackage[utf8]{inputenc}
\usepackage[english]{babel}
\usepackage[T1]{fontenc}
\usepackage{amsmath}
\usepackage{hyperref}
\usepackage{bm}
\usepackage{tikz}
\usepackage{lipsum}
\usepackage{amssymb}
\usepackage{braket}
\usepackage{amsbsy}
\usepackage{physics}
\usepackage{dsfont}
\usepackage{amsthm}
\theoremstyle{definition}
\newtheorem{definition}{Definition}[section]
\newtheorem{theorem}{Theorem}[section]
\newtheorem{lemma}[theorem]{Lemma}
\newtheorem{claim}[theorem]{Claim}
\newcommand*{\vertbar}{\rule[-1ex]{0.5pt}{2.5ex}}
\newcommand*{\bigvertbar}{\rule[-1ex]{0.5pt}{4.5ex}}
\newcommand*{\horzbar}{\rule[.5ex]{2.5ex}{0.5pt}}
\newcommand*{\bighorzbar}{\rule[.5ex]{4.5ex}{0.5pt}}
\usepackage{bm}
\usepackage[shortlabels]{enumitem}
\usepackage{float}
\usepackage[numbers,sort&compress]{natbib}
\begin{document}

\title{Quantum Universally Composable Oblivious Linear Evaluation}

\author{Manuel B. Santos}
\affiliation{Instituto de Telecomunica\c c\~oes, Av. Rovisco Pais 1, 1049-001 Lisboa, Portugal}
\affiliation{Departamento de Matem\' atica, Instituto Superior T\' ecnico,
	Universidade de Lisboa, Av. Rovisco Pais 1, 1049-001 Lisboa, Portugal}
\email{manuel.batalha.santos@gmail.com}
\author{Paulo Mateus}
\affiliation{Instituto de Telecomunica\c c\~oes, Av. Rovisco Pais 1, 1049-001 Lisboa, Portugal}
\affiliation{Departamento de Matem\' atica, Instituto Superior T\' ecnico,
	Universidade de Lisboa, Av. Rovisco Pais 1, 1049-001 Lisboa, Portugal}
\email{pmat@tecnico.ulisboa.pt}
\author{Chrysoula Vlachou}
\affiliation{Instituto de Telecomunica\c c\~oes, Av. Rovisco Pais 1, 1049-001 Lisboa, Portugal}
\affiliation{Departamento de Matem\' atica, Instituto Superior T\' ecnico,
	Universidade de Lisboa, Av. Rovisco Pais 1, 1049-001 Lisboa, Portugal}
\email{chrysoula.vlachou@tecnico.ulisboa.pt}
\maketitle

\begin{abstract}
 Oblivious linear evaluation is a generalization of oblivious transfer, whereby two  distrustful parties obliviously compute a linear function, $f(x) = ax + b$, i.e., each one provides their inputs that remain  unknown to the other, in order to compute the output $f(x)$ that only one of them receives. From both a structural and a security point of view, oblivious linear evaluation is fundamental for arithmetic-based secure multi-party computation protocols.  In the  classical case, oblivious linear evaluation protocols can be generated using oblivious transfer, and their quantum counterparts can, in principle, be constructed as straightforward extensions using quantum oblivious transfer. Here, we present the first, to the best of our knowledge, quantum protocol for oblivious linear evaluation that, furthermore, does not rely on quantum oblivious transfer. We start by presenting a semi-honest protocol and then  extend it to the dishonest setting employing a \textit{commit-and-open} strategy. Our protocol uses high-dimensional quantum states to obliviously compute   $f(x)$ on Galois Fields of prime and prime-power dimension. These constructions  utilize the existence of a complete set of mutually unbiased bases in prime-power dimension  Hilbert spaces and their linear behaviour upon the Heisenberg-Weyl operators.  We also generalize  our protocol to achieve vector oblivious linear evaluation, where several instances of  oblivious linear evaluation are generated, thus making the protocol more efficient. We prove the protocols to have static security in the framework of quantum universal composability.
\end{abstract}

\section{Introduction}\label{sec:Intro}

Oblivious Linear Evaluation (OLE) is a cryptographic task that permits two distrustful parties, say Alice and Bob, to jointly compute the output of a linear function $f(x)=ax+b$ in some finite field, $\mathbb{F}$. Alice provides inputs $a, b\in\mathbb{F}$ and Bob provides $x\in\mathbb{F}$, while the output, $f(x)$, becomes available only to Bob. As the parties are distrustful, a secure OLE protocol should not permit Alice to learn anything about Bob's input, while also Alice's inputs should remain unknown to Bob.  OLE can be seen as a generalization of Oblivious Transfer (OT) \cite{Rabin81}, a basic primitive for secure two-party computation, which, in turn, is a special case of secure multi-party computation \cite{Goldreichbook04,CCD88,Canetti00MPC}. OT has been shown to be complete  for secure multi-party computation \cite{Kilian}, i.e., any such task, including OLE, can be achieved given an OT implementation. 
A compelling reason to study OLE protocols is that they can serve as building blocks for the secure evaluation of arithmetic circuits \cite{AIK11,DKMQ12,GNN17,DGNBNT17}, just like OT allows the secure evaluation of boolean circuits \cite{GMW87}. Specifically, OLE can be used to generate multiplication triples which are the basic tool for securely computing multiplication gates \cite{DGNBNT17}.  Besides that, OLE has applications in more tasks for two-party secure computation  \cite{IPS09,ADINZ17,BCGI18,HIMV19,CDIKLOV19} and  Private Set Intersection \cite{GN19}.

Impagliazzo and Rudich  proved that OT protocols require public-key cryptography and cannot just rely on symmetric cryptography \cite{IR89}. Consequently, OLE cannot rely on symmetric cryptography either, and we need to resort to public-key cryptography.  However, Shor's  quantum algorithm \cite{Shor94}  poses a threat to the currently deployed public-key systems, motivating the search for protocols secure against quantum attacks. Bennet et al. \cite{BBCS92} and Cr{\'e}peau \cite{C94} proposed the first protocols for Quantum OT (QOT). The no-go theorems by Lo and Chau \cite{LC97,LC98} and, independently, by Mayer \cite{Mayer97}, implied that an unconditionally secure QOT without any additional assumption is impossible. 
A decade later, Damg\r{a}rd et al. proved the security of QOT in the stand-alone model under additional assumptions \cite{DFLSS09}, and Unruh \cite{Unruh10} showed that it is statistically secure and universally composable  with access only to ideal commitments.
As far as quantum OLE (QOLE) is concerned, to the best of our knowledge, no protocol has been proposed as of now.
Analogously to the classical case it is expected that one can implement QOLE based on QOT protocols. That said, in this work we propose a protocol for QOLE that, additionally, does not rely on any QOT implementation, and as such, it might provide efficiency advantages for various tasks that can be achieved directly with OLE, e.g. the evaluation of arithmetic circuits.

OLE is commonly generalized to Vector OLE (VOLE). In this setting, Alice defines a set of $k$ linear functions $(\bm{a}, \bm{b})\in\mathbb{F}^k\times\mathbb{F}^k$ and Bob receives the evaluation of all these functions on a specified element $x\in\mathbb{F}$, i.e. $\bm{f}:=\bm{a} x+ \bm{b}$. One can think of VOLE as the arithmetic analogue of string OT and show how it can be used  in certain Secure Arithmetic Computation and Non-Interactive Zero Knowledge proofs \cite{BCGI18}. Ghosh et. al  put further in evidence the usefulness of VOLE by showing that it serves as the building block of Oblivious Polynomial Evaluation \cite{GNN17}, a primitive which allows more sophisticated applications, such as password authentication, secure  list intersection,  anonymous complaint boxes \cite{NP06}, anonymous initialization for secure metering of client visits in servers \cite{NP99},  secure Taylor approximation of relevant functions (e.g. logarithm) \cite{LP02}, secure set intersection \cite{H18} and distributed generation of RSA keys \cite{G99}.  We also show how our QOLE protocol can be adapted to achieve secure VOLE.
\subsection{Contributions}\label{subsec:Intro_contributions}
We present a quantum protocol for OLE with quantum universally composable security (Definition \ref{def:statisticalquc}) in the $\mathcal{F}_{\textbf{COM}}-$hybrid model, i.e. when assuming the existence of a commitment functionality, $\mathcal{F}_{\textbf{COM}}$ (Figure \ref{fig:func_com}). To obtain a secure protocol, we take advantage of the properties of Mutually Unbiased Bases (MUBs) in high-dimensional Hilbert spaces with prime and prime-power dimension. Such a choice is motivated by recent theoretical and experimental advances that pave the way for the development and realization of new solutions for quantum cryptography \cite{CBKG02,AGS03,AKBH07,SS10,DEBZ10,Zhongetal2015,Sitetal17,Bouchardetal18,BHVBFHM18,DHMPPV21}. 
To the best of our knowledge, our protocol is the first proposal of a QOLE protocol which we prove to be quantum-UC secure. Moreover, it is not based on any QOT implementation, according to the standard approach. To prove its security, the only assumption we make is the existence of a classical commitment functionality.
We consider the static corruption adversarial model with both semi-honest and dishonest adversaries (Section \ref{subsec:Prelim_QUC}).  Finally, we modify the proposed protocol  to  generate quantum-UC secure VOLE.

\

\noindent\textbf{Main tool.} The proposed protocol $\pi_{\textbf{QOLE}}$ (Figure \ref{fig:fullprotocol}) is based on the fact that in a Hilbert space of dimension $d$ (isomorphic to $\mathbb{Z}_d$) there exists a set of MUBs $\{\ket{e^x_r}\}_{x, r\in\mathbb{Z}_d}$, such that, upon the action of a certain operator $V^b_a$,  each basis element $r$ is shifted by some linear factor $ax - b$ inside the same basis $x$:
\begin{align}
V^b_a \ket{e^x_r} = c_{a,b,x, r} \ket{e^x_{ax - b + r}},
\label{eq:main_equation_1.3.}
\end{align}
where $a, b, x, r \in \mathbb{Z}_d =\{0,1,\ldots,d-1\}$. If Alice controls the operator $V^b_a$ and Bob controls the quantum state $\ket{e^x_r}$, they are able to compute a linear function $f(x) = ax - b$ where effectively Alice controls the function $f = (a, b)$ and Bob controls its input $x$. Moreover, since Bob controls $x$ and $r$, he can receive $f(x)$ by measuring the output element. 

\

\noindent\textbf{Protocol overview.} In a nutshell, $\pi_{\textbf{QOLE}}$ (Figure \ref{fig:fullprotocol}) with inputs $f = (a,b)$ from Alice and $x$ from Bob is divided into two main phases. In the first \textit{quantum phase}, Alice and Bob use high-dimensional quantum states to generate $n$ random weak OLE (RWOLE) instances, where $n$ is the security parameter.  In this phase, Alice  outputs  $n$ random elements $f^0_i = (a^0_i, b^0_i)$, and Bob  outputs $n$ elements $(x^0_i, y^0 = f^0_i(x^0_i))$. These instances are considered to be weaker because Bob is allowed to have some amount of information about the $n$ outputs of Alice $(a^0_i, b^0_i)$. In the second \textit{post-processing phase}, Alice and Bob use classical tools to extract one secure OLE from the aforementioned $n$ instances.

More specifically, in the quantum phase, Bob randomly generates $m=(1 + t)n$ quantum states $\ket{e^{x^0_i}_{r_i}}$ and sends them to Alice. Then, Bob commits to his choice $(x^0_i, r_i)$, $\forall i\in [m]$, where for any $l\in\mathbb{N}$, $[l]$ denotes the set $\{1, \ldots, l\}$, using an ideal commitment functionality, $\mathcal{F}_{\textbf{COM}}$, and Alice asks to verify  a subset $T$ of size $tn$ of these commitments. This intermediate \textit{commit-and-open} step allows Alice to test Bob's behaviour and ensure that he does not deviate \textit{too much} from the protocol, and it is a common method used in security proofs of QOT protocols \cite{Unruh10, DFLSS09}. If Bob passes all the tests, Alice randomly generates $(a^0_i, b^0_i)$ and applies $V^{b^0_i}_{a^0_i}$ to the remaining $n$ received states $\ket{e^{x^0_i}_{r_i}}$,  for $i\in [m]\setminus T$.  For the rest of this section we relabel and denote $[n]=[m]\setminus T$. According to the expression~\eqref{eq:main_equation_1.3.}, the output states are given by $\ket{e^{x^0_i}_{a^0_i x^0_i - b^0_i + r_i}}$ and she sends them to Bob, who outputs $y^0_i = a^0_i x^0_i - b^0_i$ by measuring the received states in the corresponding basis $ x^0_i$ and subtracting $r_i$,  $\forall i\in [n]$.

The post-processing phase uses two subprotocols: a derandomization step (Figure \ref{fig:nOLE}) and an extraction step (Figure \ref{fig:privacy_amplification}). The derandomization step is based on the protocol $\pi^n_{\text{OLE}}$ from \cite{DHNO19} and transforms the $n$ RWOLE instances into $n$ weak OLE (WOLE) instances with inputs $(a_i, b_i)_{i\in [n]}$ chosen by Alice and inputs $x_i$ for $i\in [n]$ chosen by Bob. The extraction protocol uses the so-called \textit{Multi-linear Modular Hashing} family, $\mathfrak{G}$, of two-universal hash functions \cite{HK97} to render Bob's information on Alice's system useless and to extract one secure OLE out of $n$ instances of WOLE. In the extraction phase, Alice samples a two-universal hash function $g_{\bm{\kappa}}$ from $\mathfrak{G}$ and sends it to Bob. Then, with adequately-crafted vectors $(\bm{a}, \bm{b}) = \big( (a_1, \ldots, a_n), (b_1, \ldots, b_n) \big)$, Alice has $a = g_{\bm{\kappa}}(\bm{a})$ and $b = g_{\bm{\kappa}}(\bm{b})$, and Bob outputs $y = g_{\bm{\kappa}}(\bm{y})$, where $\bm{y} = \bm{a} \bm{x} + \bm{b}$ after point-wise vector multiplication with the constant vector $\bm{x} = (x, \ldots, x)$.

\

\noindent\textbf{Quantum-UC security.} 
Due to the quantum nature of the states $\ket{e^{x^0_i}_{r_i}}_{i\in [n]}$, a dishonest Alice is not able to distinguish which bases $x^0_i, i\in [n]$ are used by Bob. From her point of view, Bob's states are maximally mixed and therefore completely hide $x^0_i$. This is enough to ensure that, in the derandomization step, Alice does not receive any information about Bob's final input $x$. For a dishonest Bob, to correctly pass all Alice's tests, it means he did not cheat at all rounds with overwhelming probability. This ensures that he  has some \textit{bounded} information on Alice's random elements $(a^0_i, b^0_i)_{i\in [n]}$, and using privacy amplification techniques in the extraction step, Alice can guarantee that Bob's information about her final input $(a,b)$ is the same as in the case of an ideal OLE functionality, i.e. the probability distribution of $a$ is close to uniform.

Turning this intuition into a quantum-UC security proof requires some additional insights. First, we need a way to quantify Bob's information on Alice's elements $(a^0_i, b^0_i)$ after the testing phase and the application of the corresponding $V^{b^0_i}_{a^0_i}$  operators, for $ i\in [n]$; for this purpose we use the \textit{min-entropy} (Definition \ref{def:Hmin}). We follow the approach of \cite{DFLSS09} to guarantee that Bob does not significantly deviate from the protocol in all the rounds, and we use Theorem 1 from~\cite{Dupuis2015} to compute a concrete lower bound of Bob's min-entropy on Alice elements $(a^0_i, b^0_i)_{i\in [n]}$. Along with Lemma \ref{lem:leftover}, we have that $a = g_{\bm{\kappa}}(\bm{a})$ is close to uniform, which is sufficient to prove that Bob does not know more about $(a,b)$ than what the output $y = ax + b$ reveals. 

In order to show that  $\pi_{\textbf{QOLE}}$ is quantum-UC secure, we need to show that an ideal execution of $\pi_{\textbf{QOLE}}$ with access to $\mathcal{F}_{\textbf{OLE}}$ (Figure \ref{fig:func_ole}) is indistinguishable from a real execution of the protocol from the point of view of an external entity called the \textit{environment}. To prove this indistinguishability, we have to build a simulator that simulates the execution of the protocol in the ideal setting and generates messages on behalf of the honest simulated parties,  while trying to extract the dishonest party's inputs and feed them in $\mathcal{F}_{\textbf{OLE}}$. In particular, for a dishonest Alice, we have to demonstrate the existence of a simulator, $\mathcal{S}_A$, that generates messages on behalf of honest Bob and extracts Alice's input $(a,b)$ which, in turn, feeds into $\mathcal{F}_{\textbf{OLE}}$.  To this end, we consider that $\mathcal{S}_A$ simulates an attack by Bob at all rounds, $i$, of the protocol  which allows to extract the $m$ values of Alice  $(a^0_i,b^0_i)$. However, the commit-and-open scheme described above is designed to catch such an attack, and to work around this issue we substitute the ideal commitment functionality, $\mathcal{F}_{\textbf{COM}}$, with a fake commitment functionality, $\mathcal{F}_{\textbf{FakeCOM}}$, that allows $\mathcal{S}_A$ to open the commitments later \cite{Unruh10}.  From the remaining $n$ values $(a^0_i,b^0_i)$, $\mathcal{S}_A$ computes Alice's input $(a,b)$ and feeds it to $\mathcal{F}_{\textbf{OLE}}$.

For a dishonest Bob, we have to show the existence of a simulator, $\mathcal{S}_B$, that generates messages on behalf of  honest Alice and extracts Bob's input $x$. We assume that $\mathcal{S}_B$ has full control over  $\mathcal{F}_{\textbf{COM}}$, which means that it has access to Bob's $m$ committed values $(x^0_i, r_i)$;  the input $x$ can be easily extracted from these values. 

\

\noindent\textbf{Protocol generalizations.}  We start by generalizing   the main relation (\ref{eq:main_ingredient}) to Galois Fields of prime-power dimension, $GF(d^M) \text{ for }M>1$. Then, we show how we can obtain a protocol for quantum VOLE. In particular, from $n$ WOLE instances, we are able to generate a VOLE with size proportional to $n$, and we  bound this proportion by the min-entropy value on the WOLE instances.

\subsection{Organization}\label{subsec:Intro_organization}
In Section \ref{sec:Prelim}, we introduce notation and state results and definitions, that we will use in the rest of the paper. In Section \ref{sec:insecureQOLE}, in order to build some intuition, we present a QOLE protocol that is secure only if we consider Bob to be semi-honest; in case Bob is dishonest, its security is compromised. In Section \ref{sec:secureQROLE_overview}, we construct our secure QOLE protocol, $\pi_{\textbf{QOLE}}$, and in Section \ref{sec:secureQROLE_protocol}, we prove its security in the quantum-UC framework.   Finally, in the last Section \ref{sec:protgeneral}, we show how to generalize the presented QOLE protocol to  Galois Fields of prime-power dimensions and we also present a quantum-UC secure protocol achieving VOLE. In the last Section \ref{sec:outlook}, we give an outlook of our results and point out directions of future work. 

\section{Preliminaries}\label{sec:Prelim}
The elements  of a set, $\mathcal{X}$, are denoted by lower case letters, $x$, while capital letters, $X$, are used for random variables  with probability distribution, $P_X$, over the set $\mathcal{X}$.  For a prime number, $d$, the set $\mathbb{Z}_d= \{0, ..., d-1\}$ denotes the finite field with dimension $d$.     Vectors $\bm{v}\in \mathbb{Z}^n_d, \text{ for any } n\in\mathbb{N}$ are denoted in lower case bold letters, while capital bold letters are used for vectors of random variables. Also, we denote the uniform distribution over $\mathcal{X}$ as $\tau_{\mathcal{X}} = \frac{1}{|\mathcal{X}|}\sum_{x\in\mathcal{X}}\ketbra{x}$, where $|\mathcal{X}|$ is the size of $\mathcal{X}$. We write as~ $x\leftarrow_\$ \mathcal{X}$ the operation of assigning to $x$ an element uniformly chosen from a set $\mathcal{X}$. Finally, $\mathds{1}$ is the identity operator.

We consider Hilbert spaces, $\mathcal{H}$, of  dimension $d\geq 2$, where $d$ is a prime number.  Each Hilbert space has an orthonormal basis $\{\ket{x} : x \in \mathbb{Z}_d\}$, called the computational basis. Given different quantum systems $\mathcal{H}_1, \ldots, \mathcal{H}_n$, we use the tensor product notation to describe the joint system as $\mathcal{H}_1 \otimes \ldots \otimes \mathcal{H}_n$. The vectors in this joint system are denoted by $\ket{\bm{x}} = \ket{x_1} \otimes \ldots \otimes \ket{x_n}$, where $\bm{x} \in \mathbb{Z}_d^n$. 
Quantum states are, in general, represented as positive semi-definite operators with unitary trace acting on a Hilbert space $\mathcal{H}$, and we  denote the respective set as $\mathcal{P}(\mathcal{H})$. They are called mixed states, as opposed to pure quantum states that are represented by vectors in $\mathcal{H}$.

Let $\rho, \sigma \in \mathcal{P}(\mathcal{H})$ be two quantum states. Then, the \textit{trace distance} between them  is defined as 
\begin{align}
\delta(\rho,\sigma):=\frac{1}{2}||\rho-\sigma||_1,
\end{align}
where $||\cdot||_1$ is the $1-$Schatten norm in the space of bounded operators acting on $\mathcal{H}$.

Physically possible quantum operations are described by Completely Positive Trace Preserving (CPTP) maps. By definition, CPTP maps preserve the normalization of quantum states (being Trace Preserving) and map positive operators to positive operators (being Completely Positive), ensuring that density operators are mapped to density operators.

As we mentioned before, we use the min-entropy, $ H_{\min}$, to quantify the amount of information that Bob might obtain about Alice's system, $A$, given some (possibly quantum) side information encoded in a system $B'$, and which he can obtain by means of an attack. 
\begin{definition}
	Let $\rho_{A B'} \in \mathcal{P}(\mathcal{H}_A \otimes \mathcal{H}_{B'})$ and $\sigma_{B'} \in \mathcal{P}(\mathcal{H}_{B'})$. The \textit{min-entropy} of $\rho_{A B'}$ relative to $\sigma_{B'}$ is given by	
	\begin{align}
	H_{\min}(A | B')_{\rho|\sigma} =-\log \min\{ \lambda : \lambda \cdot \text{id}_A \otimes \sigma_{B'} \geq \rho_{A B'} \},
	\end{align}
	and 
	\begin{align}
	H_{\min}(A | B')_{\rho} = \sup_{\sigma_{B'}} H_{\min}(A | B')_{\rho|\sigma}.
	\label{def:Hmin}
	\end{align}
\end{definition}

Throughout our analysis we will also use the so-called $d$\textit{-ary entropy function}, the generalization of the standard binary entropy function:
\begin{definition}
	For $d\geq 2$, the \textit{d-ary entropy function} $h_d : [0,1]\rightarrow\mathbb{R}$ is given by
	\begin{widetext}
	\begin{align}
	h_d(x) = x \log_d(d-1) - x \log_d x - (1-x) \log_d (1-x).
	\end{align}
	\end{widetext}
	\label{def:q-ary}
	\end{definition}
\subsection{Mutually Unbiased Bases in $\mathcal{H}^d$}\label{subsec:Prelim_MUB}
In this section, we present the basics and some properties of MUBs in $\mathcal{H}^d$. This is the main tool that is used in our protocol. For more details about MUBs, see \cite{DEBZ10}.
\begin{definition}
	Let $\mathcal{B}_0 = \{\ket{\psi_1}, \ldots, \ket{\psi_d}\}$ and $\mathcal{B}_1 = \{\ket{\phi_1}, \ldots, \ket{\phi_d}\}$ be orthonormal bases in the d-dimensional Hilbert space $\mathcal{H}^d$. They are said to be \textit{mutually unbiased} if $| \braket{\psi_i}{\phi_j} | = \frac{1}{\sqrt{d}}$ for all $i, j\in \{1,\ldots,d\}$. Furthermore, a set $\{\mathcal{B}_0, \ldots, \mathcal{B}_m\}$ of orthonormal bases on $\mathcal{H}^d$ is said to be a set of MUBs if, for every $i\neq j$, $\mathcal{B}_i$ is mutually unbiased with $\mathcal{B}_j$.
\end{definition}
MUBs are extensively used in quantum cryptography because, in some sense, these bases are as far as possible from each other and the overlap between two elements from different bases is constant. Let $\big\{\ket{0}, \ldots, \ket{d-1}\big\}$ be the computational basis of $\mathcal{H}^d$, where $d$ is a prime number, and $\big\{\ket{\tilde{0}}, \ldots, \ket{\widetilde{d-1}}\big\}$ be the dual basis which is given by the Fourier transform on the computational basis:
\begin{align}
\ket{\tilde{j}} = \frac{1}{\sqrt{d}} \sum_{i=0}^{d-1} \omega^{-i j} \ket{i},\label{eq:dualbasis}
\end{align}
where $\omega = e^{\frac{2\pi i}{d}}$. We can easily verify that the computational basis and its dual basis are mutually unbiased, and we will make use of the following two operators, $V^0_a$ and $V^b_0$, to encode Alice's functions during the first (quantum) phase of the protocol.
\begin{definition}[Shift operators]
	The \textit{shift operator} $V^0_a$ shifts the computational basis by $a$ elements, i.e.
	\begin{align}
	V^0_a\ket{i} = \ket{i+a}.\label{eq:shiftoperato}
	\end{align}
	Similarly, the \textit{dual shift operator} $V^b_0$ shifts the dual basis by $b$ elements, i.e.
	\begin{align}
	V^b_0\ket{\tilde{j}} = \ket{\widetilde{j-b}}.\label{eq:dualshift}
	\end{align}
\end{definition}
The operators $V^0_a$ and $V^b_0$ are diagonal in the dual and computational basis, respectively\footnote{Note that $V^0_a$ and $V^b_0$ can be seen as a generalization of the Pauli $X$ and $Z$ operators, respectively.}, i.e.
\begin{align}
V^0_a = \sum_{j=0}^{d-1} \omega^{a j} \ketbra{\tilde{j}}{\tilde{j}} \text{ and }V^b_0 = \sum_{i=0}^{d-1} \omega^{b i} \ketbra{i}{i}.
\end{align}
Furthermore, we can define
\begin{align}
	V^b_a := V^b_0 V_a^0= \sum^{d-1}_{l=0} \omega^{(l+a)b} \ketbra{l+a}{l},
\end{align}
obtaining the so-called \textit{Heisenberg-Weyl operators}. These operators form a group of unitary transformations with $d^2$ elements; the group has $d+1$ commuting abelian subgroups of $d$ elements, and for each abelian subgroup, there exists a basis of joint eigenstates of all $V^b_a$ in the subgroup. These $d+1$ bases are pairwise mutually unbiased.
Let $x\in\mathbb{Z}_{d+1}$ label the abelian subgroups, let $l\in\mathbb{Z}_d$ label the elements of each subgroup, and let $U^x_l$ denote the corresponding subgroup operators. Finally, let the $i-$th basis element associated with the $x-$th subgroup be denoted by $\ket{e^x_i}$. Then, it can be seen that \cite{DEBZ10},
\begin{align}
U^x_l = \sum^{d-1}_{i=0} \omega^{i l} \ketbra{e^x_i}{e^x_i},
\end{align}
and \begin{align}
\ket{e^x_i} = \frac{1}{\sqrt{d}} \sum^{d-1}_{l=0} \omega^{-x l(d-l)/2 - i l}\ket{l},
\end{align}
where
\begin{align}
U^x_l = \alpha^x_l V^{x l}_l \text{ with }\alpha^x_l = \omega^{-x l(l+d)/2}.
\end{align}
One can show that
\begin{align}
V^b_a \ket{e^x_0} = c_{x,a,b} \ket{e^x_{ax - b}},\, x\in\mathbb{Z}_d, 
\end{align}
and\begin{align}
V^b_a \ket{e^d_0} = c_{d,a,b} \ket{e^d_a} {\text{ for } x=d},
\end{align}
or more generally
\begin{align}
V^b_a \ket{e^x_r} = c_{a, b, x, r} \ket{e^x_{ax - b + r}} \text{ with } c_{a, b, x, r} = \omega^{a(r + x - b) + \frac{a(a-1)}{2}x}. \label{eq:main_relation}
\end{align}

This last property is the main ingredient for the construction of our protocol as it encodes a linear evaluation based on values $a$, $b$ and $x \in \mathbb{Z}_d$\footnote{While $x \in \mathbb{Z}_{d+1}$, henceforth we consider $x \in \mathbb{Z}_{d}$, since we only use $d $ out of the $d+1$ MUBs.}.  In our protocol, we take $a,b$ -- that determine the operators $V^b_ a$ -- to be Alice's inputs and $x$ to be Bob's input.

Finally, let us see how the operators $V^b_a$ act on the so-called \textit{generalized Bell states}, since Bob's attack to the protocol is based on that. We start with the definition of the \textit{seed} Bell state 
\begin{align}
\ket{B_{0,0}} = \frac{1}{\sqrt{d}}\sum_{i} \ket{i^*, i},\label{eq:Bellseed}
\end{align}
where the map $\ket{\psi}\rightarrow\ket{\psi^*}$ is defined by taking the complex conjugate of the coefficients:
\begin{align}
\ket{\psi} = \sum_i \beta_i \ket{i}\rightarrow \ket{\psi^*} = \sum_i \beta_i^* \ket{i}.
\end{align}
Using the properties of the operators $V_a^b$, we can derive the rest of the generalized Bell states from the seed state, as
\begin{align}
\ket{B_{a,b}} = (\mathds{1} \otimes V^b_a) \ket{B_{0,0}}= \frac{1}{\sqrt{d}} \sum_{i=0}^{d-1} \omega^{(i+a)b} \ket{i^*, i+a}, \label{eq:bob_attack}
\end{align}
and one can prove that the set $\{\ket{B_{a, b}}\}_{(a,b) \in\mathbb{Z}_d^2}$ constitutes an orthonormal maximally entangled basis in the Hilbert space of two-qudit states \cite{DEBZ10}.

\subsection{Two-Universal Functions} \label{subsec:Prelim_twoUniversalFunctions}
During the  post-processing phase of the protocol, in order to amplify the privacy of Alice's inputs and outputs against the information leakage to Bob, we will use the so-called \textit{Multi-linear Modular Hashing} family of functions: 

\begin{definition}[Definition 2, \cite{HK97}]
	Let $d$ be a prime and let $n$ be an integer $n>0$. We define the family, $\mathfrak{G}$, of functions from $\mathbb{Z}_d^n$ to $\mathbb{Z}_d$, as 
	\begin{align}
	\mathfrak{G}:= \{ g_{\bm{\kappa}} : \mathbb{Z}_d^n\rightarrow \mathbb{Z}_d \, | \, \bm{\kappa}\in \mathbb{Z}_d^n \},
	\end{align}
	where the functions $g_{\bm{\kappa}}$ are defined for any $\bm{\kappa} = (\kappa_1,\ldots, \kappa_n) \in \mathbb{Z}_d^n$ and $x = (x_1,\ldots,x_n)$ as 
	\begin{align}
	g_{\bm{\kappa}}(\bm{x}) = \bm{\kappa}\cdot \bm{x} \mod d = \sum_{i=1}^n \kappa_i\,x_i \mod d.
	\end{align}
		\label{def:MMH}
\end{definition}
These functions are linear, since they are based on the modular inner product of vectors, therefore they preserve 
the structure of the OLE regarding the inputs and outputs, while they amplify  the privacy of Alice's elements. For the $\mathfrak{G}$ family one can prove the following theorem \cite{HK97}: 
\begin{theorem}
	The family $\mathfrak{G}$ is two-universal (see Definition \ref{def:2universal})\footnote{In fact, Halevi and Krawczyk \cite{HK97} proved a stronger result, namely that the $\mathfrak{G}$ family is \textit{$\Delta-$universal}, which is more general than two-universal.}.
\end{theorem}

\begin{definition}[Two-universal hash family]
	A family, $\mathfrak{G}$, of functions, $g$, with domain $D$ and range $R$ is called a \textit{two-universal hash family} if for any two distinct elements $w,w'\in D$ and for $g$ chosen at random from $\mathfrak{G}$, the probability of a \textit{collision} $g(w)=g(w')$ is at most $1/|R|$, where $|R|$ is the size of the range $R$. \label{def:2universal}
\end{definition}

Next, the \textit{Generalized Leftover Hash Lemma} is a relevant ingredient in our proof of security, as it ensures that, after applying a known function $g$  from a two-universal family to a random variable $X$, the resulting random variable $Z = g(X)$  is close to uniform, conditioned on some (possibly quantum) side information $E$. 

\begin{lemma}[Generalized Leftover Hash Lemma \cite{TSSR11}\footnote{This is a high-dimensional version of the Generalized Leftover Hash Lemma, that can be easily deduced by using Lemma 4 from \cite{TSSR11} with $d_A = d^l$. Note also that this is a special version, since Tomamichel et al. in \cite{TSSR11} prove it in the more general case for $\delta-$almost two-universal hash families.}]
	Let $X$ be a random variable, $E$ a quantum system, and $\mathfrak{G}$ a two-universal family of hash functions from $X$ to $\mathbb{Z}_d^l$. Then, on average over the choices of $g$  from $\mathfrak{G}$, the output $Z := g(X)$ is $\xi$-close to uniform conditioned on $E$, where
	\begin{align}
	\xi = \frac{1}{2}\sqrt{2^{l\log d - H_{\min}(X|E)}}.  
	\end{align}
	\label{lem:leftover}
\end{lemma}

\subsection{Quantum Universal Composability}\label{subsec:Prelim_QUC}
The Universal Composability (UC) framework was introduced by Canetti \cite{C20} in the classical setting and extended to the quantum setting by Unruh, and Ben-Or and Mayers \cite{Unruh04, BenOrMay04} (see also \cite{Unruh10, FS09}).  It renders strong composability guarantees as it ensures that the security of the protocol is independent of any external execution of the same or other protocols.  Both classical and quantum frameworks follow the same ideal-real world comparison structure and consider the same interactions between machines. The difference lies on the operations allowed by the quantum-UC framework, where  we are also allowed to store, transmit, and process  quantum states. 

More specifically, the quantum-UC security of some protocol $\mathcal{\pi}$ is drawn by comparing two scenarios. A real scenario where  $\mathcal{\pi}$ is executed and an ideal scenario where an ideal functionality, $\mathcal{F}$,  that carries out the same task and is defined \textit{a priori}, is executed. The comparison is done by a special machine, called the \textit{environment}, $\mathcal{Z}$, that supervises the execution of both scenarios and has access to any external information (e.g. concurrent executions of the same or any other protocol). In the two-party case that we are considering, the structure of the machines present in both scenarios is as follows: The real scenario has the environment, $\mathcal{Z}$, the adversary, $\mathcal{A}$, and the two parties Alice  and Bob, while the ideal scenario has the environment $\mathcal{Z}$, the simulator $\mathcal{S}$, the two parties Alice and Bob and the ideal functionality $\mathcal{F}$. 
Informally, we say that the protocol $\mathcal{\pi}$ is quantum-UC secure if no environment $\mathcal{Z}$ can distinguish between the execution of $\mathcal{\pi}$ in the real scenario and  the execution of the  functionality $\mathcal{F}$ in the ideal scenario. Any possible attack of the adversary $\mathcal{A}$ in the execution of $\mathcal{\pi}$ can be simulated by the simulator $\mathcal{S}$ in the ideal-world execution of $\mathcal{F}$, without any noticeable difference form the point-of-view of the environment $\mathcal{Z}$. Since the ideal functionality $\mathcal{F}$ is secure by definition, the real-world adversary is not able to extract any more information than what is allowed by the  $\mathcal{F}$. 

Let us now see the formal definition. Let $\mathcal{\pi}$ and $\rho$ be the real and ideal two-party protocols. We denote by $\text{Ex}_{\mathcal{\pi}^C, \mathcal{A}, \mathcal{Z}}$ (analogously $\text{Ex}_{\rho^C, \mathcal{S}, \mathcal{Z}}$) the output of the environment $\mathcal{Z}$ at the end of the real (ideal) execution, and by $C$ the corrupted party. 

\begin{definition}[Quantum-UC security \cite{Unruh10}]
	
	Let protocols $\pi$ and $\rho$ be given. We say that $\pi$ \textbf{statistically} quantum-UC emulates $\rho$ if and only if for every party, $C$, and for every adversary, $\mathcal{A}$, there exists a simulator, $\mathcal{S}$, such that for every environment $\mathcal{Z}$, and every $z\in\{0,1\}^*$, $n\in\mathbb{N}$
	\begin{align}
	\big|\text{Pr}[\text{Ex}_{\mathcal{\pi}^C, \mathcal{A}, \mathcal{Z}} (n, z) = 1] - \text{Pr}[\text{Ex}_{\rho^C, \mathcal{S}, \mathcal{Z}}(n, z) = 1]\big|\leq \mu(n),
	\end{align}
	where $\mu(n)$ is a negligible function, i.e., for every positive polynomial, $p(\cdot)$, there exists an integer $N_p>0$ such that for all $n> N_p$ we have $\mu(n)<1/p(n)$, and $n$ is the security parameter. We furthermore require that if $\mathcal{A}$ is quantum-polynomial-time, so is $\mathcal{S}$. Finally, if we consider quantum-polynomial-time $\mathcal{A}$ and $\mathcal{Z}$ we have \textbf{computational} quantum-UC security.
	\label{def:statisticalquc}
\end{definition}
The role of the simulator, $\mathcal{S}$, is to  simulate the execution of the protocol $\mathcal{\pi}$ in such a way that the environment $\mathcal{Z}$ is not able to distinguish between the two executions. 
In particular, since $\mathcal{S}$ does not have access to the inputs/outputs of the actual honest party, it runs a simulated honest party interacting with the environment which is acting as the adversary. Moreover, $\mathcal{S}$ controls the  dishonest party, therefore controlling their inputs to the ideal functionality $\mathcal{F}$. 
It then relies on its ability to extract the inputs  that the environment provides to the dishonest party. Then, it uses them along with the ideal functionality outputs in order to successfully generate a simulated execution that cannot be distinguished by the environment.

Finally, note that, in case the real execution of the protocol $\mathcal{\pi}$ makes use of some external functionality $\mathcal{F}_{\textbf{ext}}$, the simulator $\mathcal{S}$ can control and reprogram $\mathcal{F}_{\textbf{ext}}$ in whatever way suits best in the ideal world to produce an indistinguishable simulation of the real world.

Regarding the adversarial model, we consider both \textit{semi-honest} and \textit{dishonest} adversaries. Semi-honest adversaries (also called honest-but-curious or passive adversaries) do not deviate from the protocol and only try to passively gain extra information by looking at the exchanged messages. Dishonest adversaries may deviate arbitrarily from the protocol. We also adopt the static corruption adversarial model where the corruption of each party is done just before the execution of the protocol.

\

\noindent\textbf{Ideal functionalities.} 
 The main  functionality is the OLE functionality, $\mathcal{F}_{\textbf{OLE}}$, and it is described in Figure~\ref{fig:func_ole}.

\begin{figure}[!h]
	\centering
	\framebox[\linewidth][l]{%
		\parbox{0.95\linewidth}{%
			\begin{center}
				\textbf{$\mathcal{F}_{\textbf{OLE}}$ functionality}
			\end{center}
			
			\textbf{Alice's input:} $(a, b) \in\mathbb{Z}_d^2$;	\textbf{Bob's input:} $x \in\mathbb{Z}_d$
			
			\begin{itemize}
				\item Upon receiving input $(a,b) $ from Alice and $x$ from Bob, the functionality sends $f(x) := ax + b$ to Bob. 
			\end{itemize}
			\textbf{Alice's output:} $\bot$; \textbf{Bob's output:} $f(x)$
		}%
	}
	\caption{OLE functionality definition.}
	\label{fig:func_ole}
\end{figure}

In Figure~\ref{fig:func_vole}, we present the  functionality, $\mathcal{F}_{\textbf{VOLE}}^n$, which is similar to the former with the difference being that its inputs and outputs are vectors.

\begin{figure}[!h]
	\centering
	\framebox[\linewidth][l]{%
		\parbox{0.95\linewidth}{%
			\begin{center}
				\textbf{$\mathcal{F}_{\textbf{VOLE}}^n$ functionality}
			\end{center}
			
			\textbf{Alice's input:} $(\bm{a}, \bm{b}) \in \mathbb{Z}_d^{2n}$; 	\textbf{Bob's input:} $x \in \mathbb{Z}_d$
			
			\begin{itemize}
				\item Upon receiving input $(\bm{a}, \bm{b})$ from Alice and $x$ from Bob, the functionality sends $\bm{f}(x) := \bm{a}x + \bm{b} = (a_1 x + b_1 , \ldots, a_n x + b_n)$ to Bob. 
			\end{itemize}
			\textbf{Alice's output:} $\bot$;	\textbf{Bob's output:} $\bm{f}(x)$
		}%
	}
	\caption{VOLE functionality definition.}
	\label{fig:func_vole}
\end{figure}

As we mentioned already, we work in the so-called $\mathcal{F}_{\textbf{COM}}-$~hybrid model, where the real execution of the protocol has access to an ideal commitment functionality, $\mathcal{F}_{\textbf{COM}}$, defined in Figure~\ref{fig:func_com}. Note that our protocol $\mathcal{\pi}_{\textbf{QOLE}}$ makes several calls to $\mathcal{F}_{\textbf{COM}}$ and only opens a subset of the committed elements, therefore we use an index element $i$ to specify different instance calls. In the commitment phase, Bob sends $(\texttt{commit}, i, M)$ to the functionality and the functionality sends $(\texttt{commit}, i)$ to Alice. In the opening phase, Bob sends $(\texttt{open}, i)$ and the functionality sends $(\texttt{open}, i, M)$ to Alice.

\begin{figure}[!h]
	\centering
	\framebox[\linewidth][l]{%
		\parbox{0.95\linewidth}{%
			\begin{center}
				\textbf{$\mathcal{F}_{\textbf{COM}}$ functionality}
			\end{center}
			
			\textbf{Committer message:} $M$
			\begin{itemize}
				\item \textit{Commitment phase}. Upon receiving $( \texttt{commit}, M)$ from Bob, the functionality sends $\texttt{commit}$ to Alice. 
				\item \textit{Opening phase}. Upon receiving $\texttt{open}$ from Bob, the functionality sends $(\texttt{open}, M)$ to Alice. 
			\end{itemize}
		}%
	}
	\caption{Commitment functionality definition.}
	\label{fig:func_com}
\end{figure}

The $\mathcal{F}_{\textbf{COM}}$ functionality can be replaced by the commitment protocol $\pi_{\textbf{COM}}$ presented in \cite{C01} which is computationally UC-secure in the Common Reference String (CRS) model. As shown in \cite{CBGLM21} (Theorem 3.), the protocol $\pi_{\textbf{COM}}$ computationally quantum-UC realizes $\mathcal{F}_{\textbf{COM}}$ in the CRS model\footnote{The computational quantum UC-security of $\pi_{\textbf{COM}}$, shown in \cite{CBGLM21}, follows by lifting its classical computational UC-security (proved in \cite{C01}) to the quantum setting, assuming that the pseudorandom number generator in $\pi_{\textbf{COM}}$, is  robust against quantum adversaries.}. As a result, the resulting protocol $\pi_{\textbf{QOLE}}^{\pi_{\textbf{COM}}}$ is quantum-UC secure.

\section{Semi-honest QOLE protocol}\label{sec:insecureQOLE}

In order to build some intuition on the proposed protocol for QOLE, we start by presenting a simpler protocol (Figure~\ref{fig:SH_QOLE}) that is only secure under the semi-honest adversarial model. 

\begin{figure}[!h]
	\centering
	\framebox[\linewidth][l]{%
		\parbox{0.95\linewidth}{%
			\begin{center}
				\textbf{Semi-honest QOLE}
			\end{center}
			
			\textbf{Alice's input:} $(a, b) \in \mathbb{Z}^2_d$;	\textbf{Bob's input:} $x \in\mathbb{Z}_d$

			\begin{enumerate}
				
				\item Bob  randomly generates $r \in \mathbb{Z}_d$. He prepares and sends the state $\ket{e^x_r}$ to Alice.
				\item Alice prepares the operator $V^b_a$ according to her inputs $a$ and $b$. She then applies $V^b_a$ to Bob's state: $V^b_a \ket{e^x_r} = c_{x,a,b,r} \ket{e^x_{ax - b + r}}$. She sends the resulting state back to Bob. 
				\item Bob measures in the  basis  $x$,  subtracts $r$, and outputs the desired result $ax-b=:f(x)$.
			\end{enumerate}
			\textbf{Alice's output:} $\bot$; \textbf{Bob's output:} $f(x)$
		}%
	}
	\caption{Semi-honest QOLE protocol.}
	\label{fig:SH_QOLE}
\end{figure}

This semi-honest version leverages the properties of MUBs explored in Section~\ref{subsec:Prelim_MUB} and, in particular, the one presented in~\eqref{eq:main_relation}. As we saw, given the set of MUBs $\{ \ket{e^x_r} \}_{r\in\mathbb{Z}_d},\ \forall x\in\mathbb{Z}_d$, the operators $V^b_a$ simply permute the elements inside the basis $x$, according to a linear combination of the elements $a$, $b$, $x$ and $r$:

\begin{equation}
V^b_a \ket{e^x_r} = c_{a, b, x, r} \ket{e^x_{ax - b +r}}.
\label{eq:main_ingredient}
\end{equation}

Alice and Bob can use the above property to compute together a linear function $f(x) = ax - b$, where Alice chooses the parameters $a$ and $b$, and Bob chooses the input element $x$.   Bob starts by choosing a basis $x$ and an element $r$ therein, and prepares the state $\ket{e^x_r}$:  the basis choice $x$ plays the role of the input element $x$, and the basis element $r$ is used to enhance Bob's security against a potentially dishonest Alice. Then, he sends the  state $\ket{e^x_r}$ to Alice, who, in turn, applies on it the operator $V^b_a$ and sends back to Bob the resulting state. According to~\eqref{eq:main_ingredient}, Bob receives $\ket{e^x_{ax - b + r}}$,  measures it in the $x$ basis, and outputs the linear function evaluation $f(x) = ax-b$ by subtracting $r$. Thus, the correctness of the protocol is ensured by~\eqref{eq:main_ingredient}.

As far as the security of this protocol is concerned, we can easily see that it is secure against a dishonest Alice.  From her point of view, all the density matrices describing the several possible cases for $x = 0, \ldots, d-1$  are maximally mixed states. Therefore, she cannot know anything about the value of $x$. If, moreover, Bob is semi-honest the protocol remains secure. 

On the other hand, if Bob is dishonest and deviates from the protocol, he is able to find out Alice's inputs $a$ and $b$ with certainty. 
In~\eqref{eq:bob_attack}, we saw that the generalized Bell basis is generated by Alice's operators, $V^b_a$, i.e. $ \ket{B_{a,b}} = (1 \otimes V^b_a) \ket{B_{0,0}} $,  and Bob can make use of this property in order to extract her inputs $a$ and $b$. His attack can be described as follows:

\begin{enumerate}
	\item Bob prepares the state $\ket{B_{0,0}}$ and sends the second qudit to Alice.
	\item Alice applies her chosen operator $V^b_a$.
	\item Bob measures both qudits in the generalized Bell basis and outputs  $a,b$.
\end{enumerate}

It becomes clear that the protocol is secure only as long as Bob does not deviate from it; a dishonest Bob can break its security by performing the above attack. Therefore, we have to make sure that Bob sticks to the protocol. To achieve this, we apply a \textit{commit-and-open} scheme \cite{DFLSS09} that can be briefly described as follows:  Bob runs step 1. of the semi-honest QOLE protocol (Figure~\ref{fig:SH_QOLE}) multiple times, say $m$ in total, for multiple values of $x_i, \text{ and } r_i,\text{ for } i\in [m]$ and commits to these values by means of the functionality $\mathcal{F}_{\textbf{COM}}$ (Figure~\ref{fig:func_com}). Then, he sends these states to Alice, who, in turn, asks him to disclose his chosen  $x_i$'s and $r_i$'s for some of the $m$ instances that she chooses. $\mathcal{F}_{\textbf{COM}}$ forwards these committed values to Alice and she measures the corresponding received states in the disclosed bases. She can, thus, verify whether she got the right basis element for all the instances she chose to check. If Bob had used the Bell state $\ket{B_{0,0}}$ in one out of the $m$ instances, then the probability of Alice getting the correct result after measuring the state in the committed basis would be $\frac{1}{d}$. In other words, Bob would get caught with high probability $1-\frac{1}{d}$. Furthermore, if he chooses to attack all the instances, the probability of Alice getting correctly all the results is negligible, i.e. exponentially small in the number of instances, $m$.

\section{QOLE protocol}\label{sec:secureQROLE_overview}

Our QOLE protocol is divided in two phases: a quantum phase and a classical post-processing phase. The first phase uses quantum communication to generate several instances of OLE with random inputs. These instances may leak some information to the parties, therefore we refer to them as random weak OLE (RWOLE). The second phase is purely classical. It uses the RWOLE instances and extracts one classical OLE instance. The post-processing phase has two phases. It implements a derandomization procedure followed by an extraction phase that serves as a privacy amplification method.  Before proceeding, it is worth mentioning that we consider that neither dishonest party maliciously aborts the protocol. Indeed, in our setting, such a behaviour does not provide an advantage for learning the other party's input. The only case to abort the protocol is when honest Alice catches Bob cheating during the \textit{commit-and-open} stage.  

In what follows, we break down the protocol, show its correctness and retrieve some technical lemmas used for the security proof  against static dishonest adversaries.

\

\noindent\textbf{Notation.}  During the RWOLE phase, $\mathbf{F}_0 = (F^0_1,F^0_2 ,\ldots, F^0_n)$ is the vector whose components are the  random variables associated to Alice's functions. Each  $F^0_i$ ranges over the set of affine functions in $\mathbb{Z}_d$ such that $\Pr(F^0_i(x)=a^0_ix+ b^0_i)$ is uniform for all $i\in [n]$. We do not distinguish the set of affine functions in $\mathbb{Z}_d$ from $\mathbb{Z}_d^2$. The classical values $\mathbf{F}_0$ are saved in the Hilbert space $\mathcal{H}_{\mathbf{F}_0}$. The same holds for the derandomization phase, where $\mathbf{F}$ denotes the random variable for Alice's functions in the protocol $\pi^n_{\textbf{WOLE}}$. $\textbf{X}_0$ and $\textbf{Y}_0$ are the random variables for $\textbf{x}_0, \textbf{y}_0 \in \mathbb{Z}^{ n}_d$  in the RWOLE phase. and $\textbf{X}$ and $\textbf{Y}$ the corresponding random variables for $\textbf{x}, \textbf{y} \in \mathbb{Z}^{n}_d$  in the post-processing phase. Also, we use $A'$ and $B'$ to denote the system that a dishonest Alice and Bob, respectively, hold at the end of the execution of the protocol.

\subsection{RWOLE phase}

We now introduce the quantum phase of the proposed QOLE protocol, which we informally call the random weak OLE (RWOLE) phase. We denote by $\mathcal{\pi}^n_{\text{RWOLE}}$ the protocol that implements this RWOLE phase and we present it in Figure~\ref{fig:wrole}. 
\begin{figure}[!h]
	\centering
	\framebox[\linewidth][l]{%
		\parbox{0.95\linewidth}{%
			\begin{center}
				\textbf{Protocol $\mathcal{\pi}^n_{\text{RWOLE}}$}
			\end{center}
			
			\textbf{Parameters:} $n$, number of output  qudits; $t$, proportion of receiver test  qudits.
			
			\
			
			\textit{(Initialization Phase:)}
			\begin{enumerate}
				
				\item Bob randomly generates $m = (1+t)n$ different pairs $(x^0_i, r_i)$ and commits to them by sending (\texttt{commit}, $(i, x^0_i, r_i)$) to $\mathcal{F}_{\textbf{COM}}$. He prepares the states $\ket{e^{x^0_i}_{r_i}}_{i\in [m]}$ and sends them to Alice.

			\end{enumerate}
			\textit{(Test Phase:)}
			\begin{enumerate}
				\setcounter{enumi}{1}    
				
				\item Alice randomly chooses a subset of indices $T\subset [m]$ of size $t n$ and sends it to Bob.
				
				\item Bob sends $(\texttt{open}, i)$, $i\in T$, to $\mathcal{F}_{\textbf{COM}}$ and $\mathcal{F}_{\textbf{COM}}$ sends to Alice $(\texttt{open}, (i, x^0_i, r_i))$, $i\in T$.
				
				\item Alice measures the received  qudits in the corresponding $x^0_i$ basis for $i\in T$, and checks whether the received commitments are compatible with her measurements. In case there is no error she proceeds, otherwise she aborts.
				After the Test Phase,  we relabel and identify $[n]=[m]\setminus T$.
			\end{enumerate}
			\textit{(Computation Phase:)}
			\begin{enumerate}
				\setcounter{enumi}{5}       
				
				\item Alice randomly generates $n$ pairs $(a^0_i, b^0_i)$ and prepares  $V^{b^0_i}_{a^0_i}$ for $i\in[n]$.
				
				\item Alice applies these operators to the received states, i.e. $V^{b^0_i}_{a^0_i} \ket{e^{x^0_i}_{r_i}} = c_{x^0_i,a^0_i,b^0_i,r_i} \ket{e^{x^0_i}_{a^0_i x^0_i - b^0_i + r_i}}$, for $i\in[n]$, and sends the resulting states to Bob.
				
			\end{enumerate}
			\textit{(Measurement Phase:)}
			\begin{enumerate}
				\setcounter{enumi}{7}  
				
				\item Bob measures the received states in the basis $x^0_i$ for $i\in[n]$  and gets the states $\ket{e^{x^0_i}_{a^0_i x^0_i - b^0_i+r_i}}, i\in [n]$. Finally, he subtracts $r_i$, for $i\in[n]$ from his results.
				
			\end{enumerate}

			\textbf{Alice's output:} $(a^0_i, b^0_i)$, for $i \in [n]$.
			
			\textbf{Bob's output:} $(x^0_i, y^0_i)$, where $y^0_i=g_i(x^0_i) = a^0_i x^0_i - b^0_i$ for $i \in [n]$.
		}%
	}
	\caption{RWOLE protocol.}
	\label{fig:wrole}
\end{figure}

If both parties are honest the  protocol is correct: if Alice is honest, her functions $\mathbf{F}_0$ are chosen uniformly at random, and if Bob is honest he will obtain $\ket{e^{x^0_i}_{a^0_ix^0_i - b^0_i + r_i}}_{i\in [n]}$ according to \eqref{eq:main_ingredient}.

\noindent\textbf{Security.} In the case of a dishonest Alice, it is straightforward to verify that the security property of the semi-honest protocol still holds,  following the same reasoning.  In the case of dishonest Bob, though, these 
OLE random instances might leak information on Alice's random functions $\mathbf{F}_0$ to him. To quantify this  side information, we must  bound $H_{\min}(\mathbf{F}_0|B')_{\rho_{\mathbf{F}_0 B'}}$, where $\rho_{\mathbf{F}_0 B'}$ is the output state of the real execution of  $\pi^n_{\textbf{RWOLE}}$. The following lemma shows that  $\rho_{\mathbf{F}_0 B'}$  is at least $\epsilon-$close to an ideal state $\sigma_{\mathbf{F}_0 B'}$ independently of the attack that the dishonest party may perform. This ideal state $\sigma_{\mathbf{F}_0 B'}$ has a bound on $H_{\min}(\mathbf{F}_0|B')_{\sigma_{\mathbf{F}_0 B'}}$ that is proportional to the security parameter. 
\begin{lemma}[Security against dishonest Bob]
	\label{lemma:wrole_dishonest_bob}
	
	Let $\rho_{\mathbf{F}_0 B'}$ be the state given by the real execution of the protocol $\mathcal{\pi}^n_{\textbf{RWOLE}}$, where $\mathbf{F}_0$ is the system saving Alice's functions, $B'$ is Bob's (possibly quantum) system. Fix $\zeta \in ]0, 1-\frac{1}{d}]$ and let 
	\begin{align}
	\epsilon(\zeta, n) = \exp( -\frac{2 \zeta^2t^2n^2}{(nt+1)(t+1)}).
	\end{align}
	Then, for any attack of a dishonest Bob, there exists an ideal classical-quantum state $\sigma_{\mathbf{F}_0 B'}$, such that
	\begin{enumerate}
		\item $ \sigma_{\mathbf{F}_0 B'} \approx_{\epsilon} \rho_{\mathbf{F}_0 B'}$,
		\item $ H_{\min}( \mathbf{F}_0 | B' )_{\sigma_{\mathbf{F}_0 B'}} \geq \frac{n\log d}{2}(1 - h_d(\zeta)) $,
	\end{enumerate}
	where $h_d(\zeta)$ is given in Definition \ref{def:q-ary}.
\end{lemma}
The proof comprises two parts corresponding to the two conditions of Lemma \ref{lemma:wrole_dishonest_bob}: first, we prove that the state just before the \textit{Computation Phase} is close to the ideal state $\sigma_{\mathbf{F}_0 B'}$; and then, we prove that the operators applied by Alice to $\sigma_{\mathbf{F}_0 B'}$ increase the min-entropy by a specific amount that is proportional to the number of output qudits, $n$. We present the proof in Appendix  \ref{app:proofBobdishonest}, where we follow the same reasoning as Damg\r{a}rd et al. in Section 4.3 of \cite{DFLSS09}, and adapt it to our case.  We also use certain results from \cite{Dupuis2015} to establish the lower bound of property 2.

\subsection{Post-processing phase}\label{qole_protocol}

The $\mathcal{\pi}^n_{\textbf{RWOLE}}$ protocol (Figure~\ref{fig:wrole})  generates several instances of RWOLE, which leak information to Bob about Alice's inputs. In this section, we present the  post-processing phase that allows to extract one secure QOLE out of several RWOLE instances. Combining them is sufficient to generate a secure QOLE protocol, because Bob has only a negligible probability of attacking \textit{all} the weak instances without being caught; indeed, if he chooses to attack one of them the probability of Alice not aborting is $\frac{1}{t+1}+\frac{t}{d(1+t)}$, while if he chooses to attack all of them this probability becomes  $\frac{1}{d^{tn}}$, which is negligible in $n$, thus ensuring the asymptotic security of our protocol. The post-processing  comprises two subprotocols: the first is a derandomization protocol (Figure~\ref{fig:nOLE})  that integrates the randomized outputs of RWOLE into a deterministic scheme where Alice and Bob choose their inputs;   the second is an extraction protocol (Figure~\ref{fig:privacy_amplification}) that generates a secure QOLE protocol from these deterministic weak instances by means of a two-universal family of hash functions.  Note that the classical post-processing phase does not give any advantage to a potentially dishonest Alice, therefore we only need to prove security against dishonest Bob.

\subsubsection{Derandomization}
Our derandomization protocol, denoted as  $\mathcal{\pi}^n_{\textbf{WOLE}}$ and summarized in Figure~\ref{fig:nOLE}, transforms the random RWOLE instances into deterministic ones, which we informally call  weak OLE (WOLE). The output of $\pi^n_{\textbf{WOLE}}$ is still a weak version of OLE because Bob is allowed to have knowledge on Alice's inputs. The difference between RWOLE and WOLE is that the parties now  choose their inputs. 
Our derandomization protocol is an adaptation of the protocol in  \cite{DHNO19}. We denote by $*$ the product of two matrices of the same dimensions, such that the result is also a matrix of the same dimensions whose elements are the product of the respective elements of the operand matrices.

\begin{figure}[!h]
	\centering
	\framebox[\linewidth][l]{%
		\parbox{0.95\linewidth}{%
			\begin{center}
				\textbf{Protocol $\mathcal{\pi}^n_{\text{WOLE}}$}
			\end{center}
			
			\textbf{Alice's input:} $(\bm{a}, \bm{b})\in \mathbb{Z}^{2n}_d$; \textbf{Bob's input:} $\bm{x}\in \mathbb{Z}^{n}_d$
			
			\begin{enumerate}
				\item Alice and Bob run the $\mathcal{\pi}^n_{\textbf{RWOLE}}$ protocol and receive  $(\bm{a}_0, \bm{b}_0)$ and $(\bm{x}_0, \bm{y}_0)$, respectively.
				\item Bob computes and sends to Alice $\bm{c} = \bm{x} - \bm{x}_0$.
				\item Alice computes and  sends to Bob $\bm{d} = \bm{a} - \bm{a}_0$ and $\bm{s} = \bm{b}_0 + \bm{a} * \bm{c} + \bm{b}$. 
				\item Bob computes $\bm{y} = \bm{y}_0 + \bm{x} * \bm{d} - \bm{d} * \bm{c} + \bm{s}$.
			\end{enumerate}

			\textbf{Alice's output:} $\bot$; \textbf{Bob's output:} $\bm{y} = \bm{a} * \bm{x} + \bm{b}$
			}%
	}
	\caption{WOLE protocol.}
	\label{fig:nOLE}
\end{figure}

\

\noindent\textbf{Security.}  The requirements to prove security against dishonest Bob are summarized in Lemma~\ref{lemma:wole_bob_dishonest}, which is very similar in structure to Lemma~\ref{lemma:wrole_dishonest_bob}. We show that the real output state $\rho_{\mathbf{F} B'}$ of the protocol $\pi^n_{\textbf{WOLE}}$ is $\epsilon-$close to an ideal state $\sigma_{\mathbf{F} B'}$, which has min-entropy lower-bounded by a fixed value proportional to the security parameter $n$. Intuitively, this means that Bob's state is indistinguishable from a state where his knowledge on Alice's inputs is limited.

\begin{lemma}
	\label{lemma:wole_bob_dishonest}
	
	Let $\rho_{\mathbf{F} B'}$ be the state given by the real execution of the protocol $\mathcal{\pi}^n_{\textbf{WOLE}}$, where $\mathbf{F}$ is the system saving Alice's inputs, $B'$ is Bob's (possibly quantum) system. Fix $\zeta \in ]0, 1-\frac{1}{d}]$ and let 
	
	\begin{equation}
	\epsilon(\zeta, n) =\exp( -\frac{2 \zeta^2t^2n^2}{(nt+1)(t+1)}).
	\label{eq:epsilon}
	\end{equation}Then, for any attack of a dishonest Bob, there exists a classical-quantum state $\sigma_{\mathbf{F} B'}$ such that
	
	\begin{enumerate}
		\item $\sigma_{\mathbf{F} B'} \approx_{\epsilon} \rho_{\mathbf{F} B'}$,
		\item $ H_{\min}( \mathbf{F} | B' )_{\sigma_{\mathbf{F} B'}} \geq \frac{n\log d}{2}(1 - h_d(\zeta)) $,
	\end{enumerate}
	where $h_d(\zeta)$ is given in Definition \ref{def:q-ary}.
	
\end{lemma}

\begin{proof}
	Alice holds the system $A = \mathbf{F} \mathbf{F}_0 \mathbf{C} \mathbf{D} \mathbf{S}$, where $\mathbf{F} = (\mathbf{F}_{\bm{a}}, \mathbf{F}_{\bm{b}})$ refers to her inputs $(\bm{a}, \bm{b})\in \mathbb{Z}^{2n}_d$, $\mathbf{F}_0 = (\mathbf{F}_{\bm{a}_0}, \mathbf{F}_{\bm{b}_0})$ is the subsystem obtained from the RWOLE phase, and $\mathbf{C}, \mathbf{D}$ and $\mathbf{S}$ are classical subsystems used to save the values of  $\bm{c}$, $\bm{d}$, and $\bm{s}$ from the protocol, respectively.  Bob holds the system $B' = \mathbf{C} \mathbf{D} \mathbf{S} B'_0$ where $\mathbf{C}, \mathbf{D}$ and $\mathbf{S}$ are the subsystems on Bob's side where the values of $\bm{c}$, $\bm{d}$ and $\bm{s}$ are saved, respectively, and $B'_0 = \mathbf{Y}_0 E_0$ is his (possibly quantum) system generated from the RWOLE phase.
	
	To prove property $1.$, we will use Lemma~\ref{lemma:wrole_dishonest_bob}, namely that the state $\rho_{\mathbf{F}_0 B'_0}$  resulting from the RWOLE scheme is $\epsilon-$close to the ideal state $\sigma_{\mathbf{F}_0 B'_0}$. Then, we will show that the operations applied to $\rho_{\mathbf{F}_0 B'_0}$ during the derandomization process can only decrease the distance between the real and the ideal output states of the WOLE protocol, thus keeping them at least $\epsilon-$close.
	We start by specifying the operators corresponding to the classical operations executed in steps $2$ and $3$ of $\mathcal{\pi}^n_{\textbf{WOLE}}$. In step 2, a dishonest Bob can send to Alice some value $\bm{c}$ that depends on his system $B'_0$. So, he starts by applying a CPTP map $\mathcal{T}_{B'_0 \rightarrow \mathbf{C} B'_0}: \mathcal{P}\left( \mathcal{H}_{B'_0}\right) \rightarrow \mathcal{P}\left( \mathcal{H}_{B'_0}\otimes \mathcal{H}_{\mathbf{C}} \right)$ to his state and then projects it into the Hilbert space $\mathcal{H}_{\mathbf{C}}$. The operator for step $2$ is a CPTP map 
	\begin{widetext}
	\begin{align}
		\mathcal{O}^{(2)} :  \mathcal{P}\left(\mathcal{H}_{\mathbf{F}_0} \otimes \mathcal{H}_{B'_0}\right) \rightarrow \mathcal{P}\left(\mathcal{H}_{\mathbf{F}_0} \otimes \mathcal{H}_{\mathbf{C}} \otimes \mathcal{H}_{\mathbf{D}} \otimes \mathcal{H}_{\mathbf{S}} \otimes \mathcal{H}_{B'_0}\right)
	\end{align}
	\end{widetext}
	described by his action on some general quantum state $\rho$, as
	\begin{widetext}
	\begin{align}
	\mathcal{O}^{(2)}(\rho) = \mathds{1} \otimes \sum_{\bm{d}, \bm{s}, \bm{c}} \ketbra{\bm{c}}_{\mathbf{C}} \mathcal{T}_{B'_0 \rightarrow \mathbf{C} B'_0}(\rho) \ketbra{\bm{c}}_{\mathbf{C}} \otimes \ketbra{\bm{d}}_{\mathbf{D}}\otimes\ketbra{\bm{s}}_{\mathbf{S}} .
	\end{align}
	\end{widetext}
	In step 3, Bob takes no action. Since Alice is honest, the operator for this step simply describes her action on subsystems $\mathbf{D}$ and $\mathbf{S}$ according to her choice at subsystem $\mathbf{F}$. This operator is a CPTP map
		\begin{align}
		\mathcal{O}^{(3)} :  \mathcal{P}\left(\mathcal{H}_{\mathbf{F}_0} \otimes \mathcal{H}_{\mathbf{C}} \otimes \mathcal{H}_{\mathbf{D}} \otimes \mathcal{H}_{\mathbf{S}} \otimes \mathcal{H}_{B'_0}\right) \rightarrow \mathcal{P}\left(\mathcal{H}_{\mathbf{F}} \otimes \mathcal{H}_{\mathbf{F}_0} \otimes \mathcal{H}_{\mathbf{C}} \otimes \mathcal{H}_{\mathbf{D}} \otimes \mathcal{H}_{\mathbf{S}} \otimes \mathcal{H}_{B'_0}\right)
		\end{align}
	described by his action on some general quantum state $\rho$, as
	\begin{align}
	\mathcal{O}^{(3)}(\rho) = \frac{1}{d^{2n}} \sum_{\bm{a}, \bm{b}} \mathcal{P}^{\bm{a}, \bm{b}}\, \rho \, (\mathcal{P}^{\bm{a}, \bm{b} })^\dagger,
	\end{align}
	where 
	\begin{widetext}
		\begin{align}
		\mathcal{P}^{\bm{a}, \bm{b}} &= \ket{\bm{a}, \bm{b}}_{\mathbf{F}} \otimes\nonumber\\ & \sum_{\bm{a}_0, \bm{b}_0, \bm{c}} \ketbra{\bm{a}_0, \bm{b}_0}_{\mathbf{F}_0} \otimes \ketbra{\bm{c}}_{\mathbf{C}} \otimes \ketbra{\bm{a} - \bm{a}_0}_{\mathbf{D}} \otimes \ketbra{\bm{b}_0 + \bm{a} \cdot \bm{c} + \bm{b}}_{\mathbf{S}}.
		\end{align}
	\end{widetext}
	Note that $\mathcal{O}^{(2)}$ adds subsystems $\mathbf{C} \mathbf{D} \mathbf{S}$ and distributes $\mathbf{C}$ according to Bob's action. The operator $\mathcal{O}^{(3)}$ adds subsystem $\mathbf{F}$ and projects $\mathbf{D} \mathbf{S}$ according to the information at subsystem $\mathbf{F} \mathbf{F}_0$ and the expressions of $\bm{d}$ and $\bm{s}$. Regarding the trace distance between the real and ideal states, we have: 
	\begin{align}
	\delta(\rho_{\mathbf{F}_0 B'_0} , \sigma_{\mathbf{F}_0 B'_0})&\geq\delta\Big( \mathcal{O}^{(2)}(\rho_{\mathbf{F}_0 B'_0}), \mathcal{O}^{(2)}(\sigma_{\mathbf{F}_0 B'_0})\Big)\nonumber \\
	&\geq \delta\Big( \mathcal{O}^{(3)} \mathcal{O}^{(2)}(\rho_{\mathbf{F}_0 B'_0}), \mathcal{O}^{(3)} \mathcal{O}^{(2)}(\sigma_{\mathbf{F}_0 B'_0})\Big) \nonumber\\ & =\delta( \rho_{\mathbf{F} B'}, \sigma_{\mathbf{F} B'}).
	\end{align}
	For the above inequalities, we took into account that $\mathcal{O}^{(2)}$ and $\mathcal{O}^{(3)}$ are CPTP maps, and as such they do not increase the trace distance (see Lemma 7 in \cite{U17}).  For the last equality, recall that $B' = \bm{C}\bm{D}\bm{S}B'_0$. Now, from  Lemma~\ref{lemma:wrole_dishonest_bob}, we have that $\sigma_{\mathbf{F}_0 B'_0} \approx_{\epsilon} \rho_{\mathbf{F}_0 B'_0}$.
	Hence, we conclude that $\delta( \rho_{\mathbf{F} B'}, \sigma_{\mathbf{F} B'}) \leq \epsilon(\zeta, n)$ for $\epsilon(\zeta, n)$ given as~\eqref{eq:epsilon}, i.e. $\sigma_{\mathbf{F} B'} \approx_{\epsilon} \rho_{\mathbf{F} B'}$.
	
	We move on to prove property $2$. Consider the bijective function $g^{\bm{c},\bm{d},\bm{s}} : \mathbb{Z}^{2n}_d \rightarrow \mathbb{Z}^{2n}_d$ given by 
	\begin{align}
	g^{\bm{c},\bm{d},\bm{s}}(\bm{x}, \bm{y}) = (\bm{x} + \bm{d}, \bm{s} - \bm{y} - (\bm{x} + \bm{d}) * \bm{c}),
	\end{align}
	for fixed $\bm{c}, \bm{d}$ and $\bm{s}$. Essentially, $g^{\bm{c},\bm{d},\bm{s}}$ describes how the input vector $(\bm{a}, \bm{b})$ is related to the RWOLE output vector $(\bm{a}_0, \bm{b}_0)$:
	\begin{align}
	\label{eq:def_ab}
	(\bm{a}, \bm{b})  = g^{\bm{c},\bm{d},\bm{s}}(\bm{a}_0, \bm{b}_0) = (\bm{a}_0 + \bm{d}, \bm{s} - \bm{b}_0 - (\bm{a}_0 + \bm{d})  * \bm{c}).
	\end{align}
	Intuitively, this means that the subsystem $\mathbf{F}$ is defined by the subsystems $\mathbf{F}_0 \mathbf{C} \mathbf{D} \mathbf{S}$. We can rewrite  the action of the operator $\mathcal{O}^{(3)}$  on some general quantum state $\rho$ as follows:
	\begin{align}
	\mathcal{O}^{(3)}(\rho) = \frac{1}{d^{2n}} \sum_{\bm{d}, \bm{s}} \mathcal{P}^{\bm{d}, \bm{s}}\, \rho \, (\mathcal{P}^{\bm{d}, \bm{s} })^\dagger,
	\end{align}
	where 
	\begin{widetext}
		\begin{align}
		\mathcal{P}^{\bm{d}, \bm{s}} = \sum_{\bm{a}_0, \bm{b}_0, \bm{c}} \ketbra{g^{\bm{c},\bm{d},\bm{s}}(\bm{a}_0, \bm{b}_0)}{\bm{a}_0, \bm{b}_0}_{\mathbf{F}} \otimes \ketbra{\bm{a}_0, \bm{b}_0, \bm{c}, \bm{d}, \bm{s}}_{\mathbf{F}_0 \bm{C} \bm{D} \bm{S}}. 
		\end{align}
	\end{widetext}

	Hence, for the min-entropy bound,  we have:
	\begin{align}
	H_{\min}(\mathbf{F} \,|\, B')_{\sigma_{\mathbf{F}B'}} &= H_{\min}(\mathbf{F}_{\bm{a}}, \mathbf{F}_{\bm{b}} \,|\, B')_{\sigma_{\mathbf{F}B'}}\nonumber\\&=H_{\min}( f^{\bm{C},\bm{D},\bm{S}}(\mathbf{F}_{\bm{a}_0}, \mathbf{F}_{\bm{b}_0} )  \,|\, \mathbf{C} \mathbf{D} \mathbf{S} B'_0)_{\mathcal{O}^{(3)} \mathcal{O}^{(2)}\sigma_{\mathbf{F}_0B'_0}} \nonumber\\
	& \geq H_{\min}(\mathbf{F}_{\bm{a}_0}, \mathbf{F}_{\bm{b}_0} \,|\, \mathbf{C} \mathbf{D} \mathbf{S} B'_0)_{ \mathcal{O}^{(2)}\sigma_{\mathbf{F}_0B'_0}} \label{first_rel} \\
	& \geq H_{\min}(\mathbf{F}_{\bm{a}_0}, \mathbf{F}_{\bm{b}_0} \,|\, B'_0)_{\sigma_{\mathbf{F}_0B'_0}} \label{second_rel}\\
	&\geq \frac{n\log d}{2}(1 - h_d(\zeta)). \label{third_rel}
	\end{align}
	The inequality \eqref{first_rel} comes from Lemma~\ref{lemma:bijectivefunction},  as $g^{\bm{c},\bm{d},\bm{s}}$ is bijective. The inequality \eqref{second_rel} comes from Theorem 6.19 in \cite{T16}, as the operator $\mathcal{O}^{(2)}$ takes the form of $\mathcal{O}^{(2)} = \mathds{1} \otimes \mathcal{M}$, where $\mathcal{M}$ is a CPTP map. The last inequality comes from Lemma~\ref{lemma:wrole_dishonest_bob}, property 2.
	\end{proof}

\subsubsection{Extraction}

In this section, we present our extraction protocol, $\mathcal{\pi}_{\textbf{EXT}}$,  that generates one OLE instance using the derandomization protocol $\mathcal{\pi}^n_{\textbf{WOLE}}$ and the two-universal family of hash functions, $\mathfrak{G}$ (see \cite{HK97} and Definition~\ref{def:MMH}). This family uses the inner product between two vectors in the $\mathbb{Z}^n_d$ space, and since OLE only involves linear operations, we can apply the inner product operation to all vectors $\bm{a}$, $\bm{b}$ and $\bm{y}$ without affecting the overall structure. The protocol $\mathcal{\pi}_{\textbf{EXT}}$ is summarized in Figure~\ref{fig:privacy_amplification} and uses the $n$ instances of WOLE in such a way that Bob's knowledge on Alice's inputs decreases exponentially with respect to the security parameter $n$\footnote{This extraction step is similar to the privacy amplification step of QKD protocols.}.

\begin{figure}[!h]
	\centering
	\framebox[\linewidth][l]{%
		\parbox{0.95\linewidth}{%
			\begin{center}
				\textbf{Protocol $\mathcal{\pi}_{\textbf{EXT}}$}
			\end{center}
			
			\textbf{Alice's input:} $(a, b) \in\mathbb{Z}_d^2$; \textbf{Bob's input:} $x \in\mathbb{Z}_d$
			
			\begin{enumerate}
				\item Alice chooses randomly some function $g_{\bm{\kappa}} \in \mathfrak{G}$ and sends it to Bob.
				\item  Alice randomly generates $a_2, \ldots, a_n, b_2, \ldots, b_n\leftarrow_{\$}\mathbb{Z}_d$. She computes $a_1 = \big(a - \sum_{i=2}^{n} a_i \kappa_i\big)/\kappa_1$ and $b_1 = \big(b - \sum_{i=2}^{n} b_i \kappa_i\big)/\kappa_1$. We write $\bm{a} = (a_1, \ldots, a_n)$ and $\bm{b} = (b_1, \ldots, b_n)$.
				\item Alice and Bob run the derandomization protocol $\mathcal{\pi}^n_{\textbf{WOLE}}((\bm{a}, \bm{b}), \bm{x})$. Bob receives  $\bm{y}$ as output.
				
				\item Bob computes $y = g_{\bm{\kappa}}(\bm{y})$.
			\end{enumerate}
			
			\textbf{Alice's output:} $\bot$; \textbf{Bob's output:} $y$
			}
	}
	\caption{Extraction protocol.}
	\label{fig:privacy_amplification}
\end{figure}

The correctness of $\mathcal{\pi}_{\textbf{EXT}}$ follows from linearity:
\begin{widetext}
\begin{align}
	y  = g_{\bm{\kappa}}(\bm{y}) & = \bm{\kappa} \cdot (a_1 x + b_1 ,\ldots a_n x + b_n) \nonumber\\
	&=\bm{\kappa} \cdot \Bigg(\frac{a - \sum_{i=2}^{n} a_i \kappa_i}{
		\kappa_1}\, x + \frac{b - \sum_{i=2}^{n} b_i \kappa_i}{\kappa_1},\, a_2 x + b_2, \ldots , a_n x + b_n \Bigg)\nonumber\\
	&=a x + b.
\end{align}
\end{widetext}
\

\noindent\textbf{Security.} 
By definition, the derandomization protocol leaks some information to Bob about Alice's inputs $(\bm{a}, \bm{b})$. Since $\bm{y = a  * x + b}$, without loss of generality, any leakage of Alice's inputs can be seen as a leakage on just $\bm{a}$. In this case, the min-entropy of $\mathbf{F} = (\mathbf{F}_{\bm{a}}, \mathbf{F}_{\bm{b}})$ should be the same as the min-entropy of $\mathbf{F}_{\bm{a}}$. Now, recall the $\mathcal{F}_{\textbf{OLE}}$ definition (Figure \ref{fig:func_ole}),  and note that Bob does not possess any knowledge about Alice's input $(a,b)$ other than what can be deduced from his input and output $(x, y)$. Similarly, since $y = ax + b$, Bob has some knowledge on the relation between $a$ and $b$ and --  as $b$ is completely determined by $(a,x,y)$ --  we only have to guarantee that $a$ looks uniformly random to Bob. The role of the hash functions used in the above protocol $\mathcal{\pi}_{\textbf{EXT}}$ is precisely to extract a uniformly random $a$ from the leaky vector $\bm{a}$, while preserving the structure of the OLE. This result is summarized in Lemma~\ref{lemma:extraction} and its proof is based on  Lemma \ref{lem:leftover}. 
\begin{lemma}
	Let $\rho_{F B'}$ be the state given by the real execution of the protocol $\pi_{\textbf{EXT}}$, where $F$ is the system saving Alice's inputs $(a,b)$, $B'$ is Bob's (possibly quantum) system. Fix $\zeta \in ]0, 1-\frac{1}{d}]$ and let 
	\begin{align}
	\epsilon(\zeta, n) = \exp( -\frac{2 \zeta^2t^2n^2}{(nt+1)(t+1)}).
	\end{align}
	Then, for any attack of a dishonest Bob, there exists a classical-quantum state  $\sigma_{F B'}$, where   $F = (F_{a}, F_{b})$, such that 
	
	\begin{enumerate}
		\item $ \sigma_{F B'} \approx_{\epsilon} \rho_{F B'}$, and
		\item  $\delta( \tau_{\mathbb{Z}_d} \otimes \sigma_{B'},\, \sigma_{F_{a} B'} ) \leq K\, 2^{-n \, f_d(\zeta)}$, where $K = \frac{\sqrt{d}}{2}$, $f_d(\zeta) = \frac{\log d}{4} (1-h_d(\zeta))$, $n$ is the security parameter, and $h_d(\zeta)$  is given in Definition \ref{def:q-ary}.
	\end{enumerate}
	\label{lemma:extraction}
\end{lemma}
\begin{proof}
	To prove property $1$, we note that the extraction operation applied to the output of $\mathcal{\pi}^n_{\textbf{WOLE}}$ can be described by a projective operator on the space  $F = (F_{a}, F_{b})$. Therefore, as in the case of  Lemma~\ref{lemma:wole_bob_dishonest}, property $1$ follows from  the fact that CPTP maps do not increase the trace distance (see Lemma 7 in \cite{U17}).
	
	Regarding property $2$,  let us first consider Bob's subsystem $E$  to integrate Bob's inputs $\bm{x}$, i.e. $E = \mathbf{X} E'$. Then, his full system $B'$ is identified with $\mathbf{Y} E=\mathbf{Y}\mathbf{X} E'$. We have: 
	\begin{align}
		H_{\min}(\mathbf{F}\, | \, \mathbf{Y} E)_{\sigma_{\mathbf{F}\mathbf{Y}E}} &= H_{\min}(\mathbf{F}_{\bm{a}}, \mathbf{F}_{\bm{b}} \, | \, \mathbf{Y} \mathbf{X} E')_{\sigma_{\mathbf{F}\mathbf{Y}E}} \\
		&= H_{\min}(\mathbf{F}_{\bm{a}}, \mathbf{Y} - \mathbf{F}_{\bm{a}} \mathbf{X} \, | \, \mathbf{Y} \mathbf{X} E')_{\sigma_{\mathbf{F}_a\mathbf{Y}E}} \\
		&= H_{\min}(\mathbf{F}_{\bm{a}}\, | \, \mathbf{Y} \mathbf{X} E')_{\sigma_{\mathbf{F}_a\mathbf{Y}E}}.
	\end{align}
	Therefore, 
	\begin{align}
	H_{\min}(\mathbf{F}_{\bm{a}}\, | \, \mathbf{Y} E)_{\sigma_{\mathbf{F}_a\mathbf{Y}E}} \geq \frac{n\log d}{2} (1 - h_d(\zeta)).
	\end{align}
	
	Now, since $\mathfrak{G}$ is a two-universal family of hash functions, we can directly apply  Lemma~\ref{lem:leftover} for $l=1$. It follows that $F_a$ is $\xi-$close to uniform conditioned on $\mathbf{Y} E$, i.e.
\begin{align}
\delta( \tau_{\mathbb{Z}_d} \otimes \sigma_{\mathbf{Y} E},\, \sigma_{F_{a} \mathbf{Y} E} ) \leq \frac{1}{2}\sqrt{2^{\log d - \frac{n\log d}{2}(1 - h_d(\zeta))}} = K\, 2^{-n \, f_d(\zeta)} =: \xi,
\end{align}
	where $K = \frac{\sqrt{d}}{2}$, $f_d(\zeta) = \frac{\log d}{4} (1-h_d(\zeta))$ and $n$ is the security parameter.
	\end{proof}

Now, we are in position to combine the above $\pi^n_{\textbf{RWOLE}}$, $\pi^n_{\textbf{WOLE}}$ and $\pi_{\textbf{EXT}}$, and present the full protocol $\pi_{\textbf{QOLE}}$ in Figure~\ref{fig:fullprotocol}. 
\begin{figure}[H]
	\centering
	\framebox[1\linewidth]{%
		\parbox{0.95\linewidth}{%
			\begin{center}
				\textbf{Protocol $\mathcal{\pi}_{\text{QOLE}}$}
			\end{center}
			\textbf{Parameters:} $n$, security parameter; $tn$, number of test qudits.
			
			\textbf{Alice's input:} $(a, b)\in\mathbb{Z}_d^2$; \textbf{Bob's input:} $x\in\mathbb{Z}_d$
			
			\vspace{0.5\baselineskip}
			
			\textit{(Quantum phase:)}
			
			\begin{enumerate}
				\item Bob randomly generates $m = (1+t)n$ different pairs $(x^0_i, r_i)\in\mathbb{Z}_d^2$ and commits to them by sending (\texttt{commit}, $(i,x^0_i, r_i)_{ i\in [m]}$) to  $\mathcal{F}_{\textbf{COM}}$. He also prepares the quantum states $\ket{e^{x^0_i}_{r_i}}_{ i\in [m]}$ and sends them to Alice. 
				
				\item Alice randomly chooses a subset of indices $T\subset [m]$ of size $t n$ and sends it to Bob.
				
				\item Bob sends $(\texttt{open}, i)_{ i\in T}$ to $\mathcal{F}_{\textbf{COM}}$ and $\mathcal{F}_{\textbf{COM}}$ sends to Alice $(\texttt{open}, (i, x^0_i, r_i))_{ i\in T}$.
				
				\item Alice measures the received quantum states in the corresponding $x^0_i$ basis for $i\in T$, and checks whether the received commitments are compatible with her measurements. She proceeds in case there is no error, otherwise she aborts.
				
				\item Alice randomly generates $n$ pairs $(a^0_i, b^0_i)\in\mathbb{Z}_d^2$ and prepares  $V^{b^0_i}_{a^0_i}$ for $i\in [m]\setminus T$. We relabel $\bm{a}_0 = (a^0_1, \ldots, a^0_n)$, $\bm{b}_0 = (b^0_1, \ldots, b^0_n)$ and $\bm{x}_0 = (x^0_1, \ldots, x^0_n)$, and from now on identify $[m]\setminus T\equiv [n]$.
				
				\item Alice  $\forall i\in [n]$ applies $V^{b^0_i}_{a^0_i}$ to the received state $\ket{e^{x^0_i}_{r_i}}$, i.e. $V^{b^0_i}_{a^0_i} \ket{e^{x^0_i}_{r_i}} = c_{x^0_i,a^0_i,b^0_i,r_i} \ket{e^{x^0_i}_{a^0_i x^0_i - b^0_i + r_i}}$, and sends the resulting states to Bob.
				
				\item Bob  $\forall i\in [n]$ measures the received state in the corresponding basis  $x^0_i$, and gets the state $\ket{e^{x^0_i}_{a^0_i x^0_i - b^0_i+r_i}}$. Finally,  $\forall i\in [n]$ he subtracts $r_i$ from his result and gets $y^0_i = a^0_i x^0_i - b^0_i$. We write $\bm{y}_0 = (y^0_1, \ldots, y^0_n)$. 
			\end{enumerate}
			
			\textit{(Post-processing phase:)}
			
			\begin{enumerate}
				\setcounter{enumi}{7} 
				\item Bob defines $\bm{x} = (x, \ldots,x)$ as the constant vector according to his input $x$.
				\item Alice chooses randomly some function $g_{\bm{\kappa}} \in \mathfrak{G}$, and she randomly generates $a_2, \ldots, a_n, b_2, \ldots, b_n\leftarrow_{\$}\mathbb{Z}_d$. She computes $a_1 = \big(a - \sum_{i=2}^{n} a_i \kappa_i\big)/\kappa_1$ and $b_1 = \big(b - \sum_{i=2}^{n} b_i \kappa_i\big)/\kappa_1$. We write $\bm{a} = (a_1, \ldots, a_n)$ and $\bm{b} = (b_1, \ldots, b_n)$.
				\item Bob computes and sends to Alice $\bm{c} = \bm{x} - \bm{x}_0$.
				\item Alice computes and sends  to Bob $\bm{d} = \bm{a} - \bm{a}_0$ and $\bm{s} = \bm{b}_0 + \bm{a}  * \bm{c} + \bm{b}$.
				\item Bob computes $\bm{y} = \bm{y}_0 + \bm{x} * \bm{d} - \bm{d}* \bm{c} + \bm{s}$.
				\item Finally, Alice sends $\bm{\kappa}$ to Bob and he computes $y = g_{\bm{\kappa}}(\bm{y})$.
			\end{enumerate}
\textbf{Alice's output:} $\bot$; \textbf{Bob's output:} $y$ 
	}}	\caption{QOLE protocol.}
	\label{fig:fullprotocol}
\end{figure}

\section{UC security}\label{sec:secureQROLE_protocol}

In this section, we will show that   $\mathcal{\pi}_{\textbf{QOLE}}$ (Figure~\ref{fig:fullprotocol}) is quantum-UC secure. More formally, we will show that $\mathcal{\pi}_{\textbf{QOLE}}$  statistically quantum-UC realizes  (Definition~\ref{def:statisticalquc}) $\mathcal{F}_{\textbf{OLE}}$ in the $\mathcal{F}_{\textbf{COM}}-$hybrid model.

\begin{theorem}[quantum-UC security of $\mathcal{\pi}_{\textbf{QOLE}}$]
	The protocol $\mathcal{\pi}_{\textbf{QOLE}}$ from Figure~\ref{fig:fullprotocol} statistically quantum-UC realizes  (see Definition~\ref{def:statisticalquc}) $\mathcal{F}_{\textbf{OLE}}$ in the $\mathcal{F}_{\textbf{COM}}-$hybrid model.
	\label{thm:QUC}
\end{theorem}
Theorem~\ref{thm:QUC} is proved by combining Lemma~\ref{lemma:dishonestAlice} and Lemma~\ref{lemma:dishonestBob} that we present below.  In the former we prove the protocol's security for the case where Alice is dishonest and Bob is honest, while in the latter we prove security in the case where Alice is honest and Bob dishonest. 
	\begin{figure}[h!]
	\centering
	\framebox[\linewidth][l]{%
		\parbox{0.95\linewidth}{%
			\begin{center}
				\textbf{Simulator $\mathcal{S}_{\text{A}}$}
			\end{center}
			
			\textit{(Quantum phase:)}
			
			\begin{enumerate}
				\item $\mathcal{S}_A$ sends \texttt{commit} to $\mathcal{F}_{\textbf{FakeCOM}}$.
				\item  $\mathcal{S}_A$ generates $ m=(1+t)n$ entangled states $\ket{B_{0,0}}_{Q_A Q_S}$ and sends subsystem $Q_A$ to Alice. 
				\item Alice asks for a set of indices $T  \subset[m] $ of size $tn$.
				\item $\mathcal{S}_A$ measures the corresponding elements of subsystem $Q_S$ using $tn$ randomly chosen bases $x^0_i$ and provides $(\texttt{open}, (i, x^0_i, r_i))$ to $\mathcal{F}_{\textbf{FakeCOM}},\  \forall i\in T$. 
				\item Upon receiving the processed system $\hat{Q}_A$ from Alice, $\mathcal{S}_A$ measures the joint system $\hat{Q}_A Q_S$ and extracts the measurement outcomes $\mathbf{F} = (\bm{a}_0, \bm{b}_0) = \big( (a^0_1,\ldots,a^0_n), (b^0_1,\ldots,b^0_n) \big)$.
			\end{enumerate}
			
			\textit{(Post-processing phase:)}
			
			\begin{enumerate}
				\setcounter{enumi}{5} 
				\item  $\mathcal{S}_A$ randomly generates a vector $\bm{c}'$ and sends to Alice.
				\item Upon receiving $\bm{d}$ and $\bm{s}$ from Alice, $\mathcal{S}_A$ extracts $\bm{a}$ and $\bm{b}$ based on its knowledge of $(\bm{a}_0, \bm{b}_0)$ as follows:
				\begin{equation}
				\begin{split}
				\bm{a} &= \bm{b} + \bm{a}_0 \\
				\bm{b} &= \bm{s} - \bm{b}_0 - \bm{a}  *\bm{c}'.
				\end{split}
				\label{eqn:extract_1}
				\end{equation}

				\item Upon receiving $\bm{\kappa}$ from Alice, $\mathcal{S}_A$ extracts her inputs $(a,b)$ as follows:
				
				\begin{equation}
				\begin{split}
				a &= \bm{a} \cdot \bm{\kappa} \\
				b &= \bm{b} \cdot \bm{\kappa}.
				\end{split}
				\label{eqn:extract_2}
				\end{equation}

				\item Finally, $\mathcal{S}_A$ sends $(a,b)$ to the ideal functionality $\mathcal{F}_{\textbf{OLE}}$.
				
			\end{enumerate} 
		}%
	}
	\caption{Simulator $\mathcal{S}_A$ against  dishonest Alice.}
	\label{fig:simulator_dis_Alice}
\end{figure}
\begin{lemma}
	The protocol $\mathcal{\pi}_{\textbf{QOLE}}$ (Figure~\ref{fig:fullprotocol}) statistically quantum-UC realizes  (see Definition~\ref{def:statisticalquc}) $\mathcal{F}_{\textbf{OLE}}$ in the $\mathcal{F}_{\textbf{COM}}-$hybrid model in the case of dishonest Alice and honest Bob.
	\label{lemma:dishonestAlice}
\end{lemma}

\begin{proof}
	
	We start by presenting the simulator $\mathcal{S}_A$ for the case where Alice is dishonest in Figure~\ref{fig:simulator_dis_Alice}. Moreover, in Figure~\ref{fig:DANetworks}, we illustrate the corresponding real and ideal networks.

\begin{figure}[h]
	\centering
	\begin{minipage}[l]{0.5\linewidth}
		\includegraphics[scale=0.3]{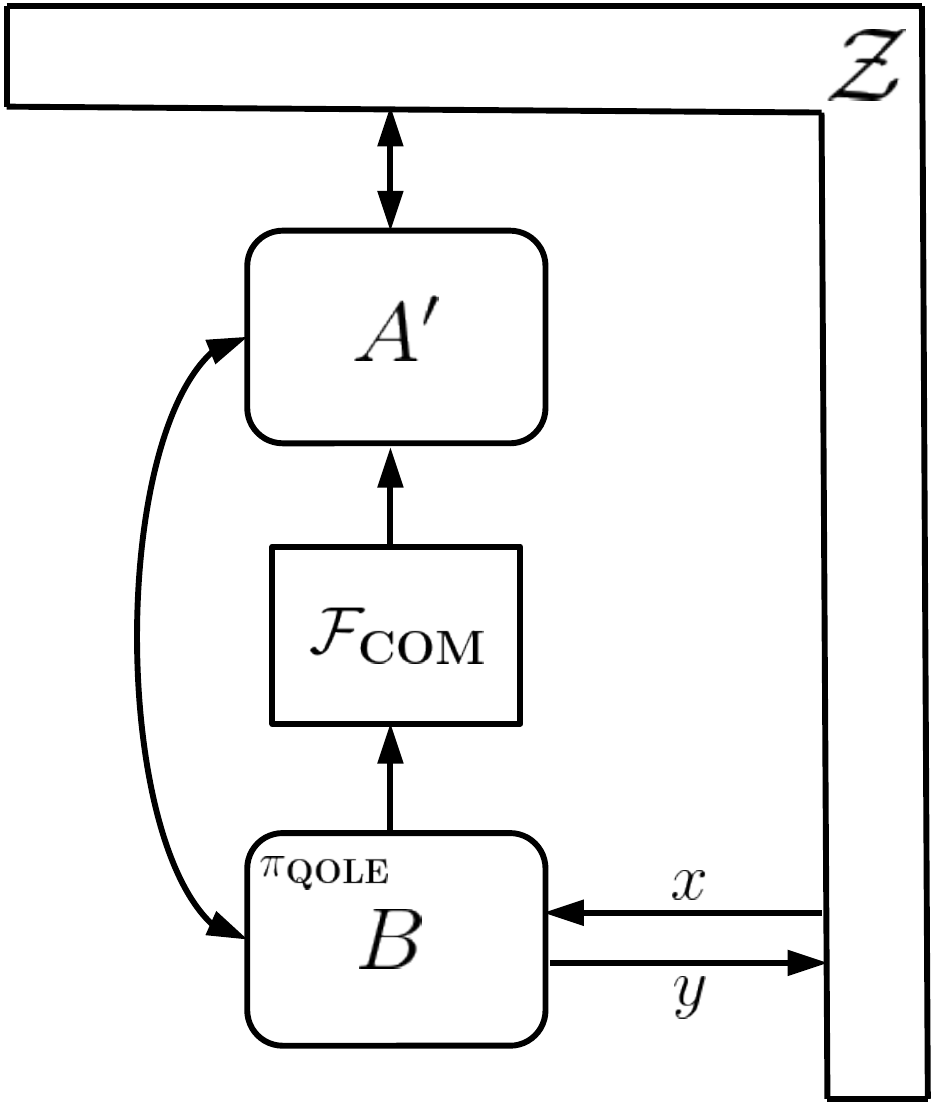}\end{minipage}\begin{minipage}[r]{0.5\linewidth}
		\includegraphics[scale=0.3]{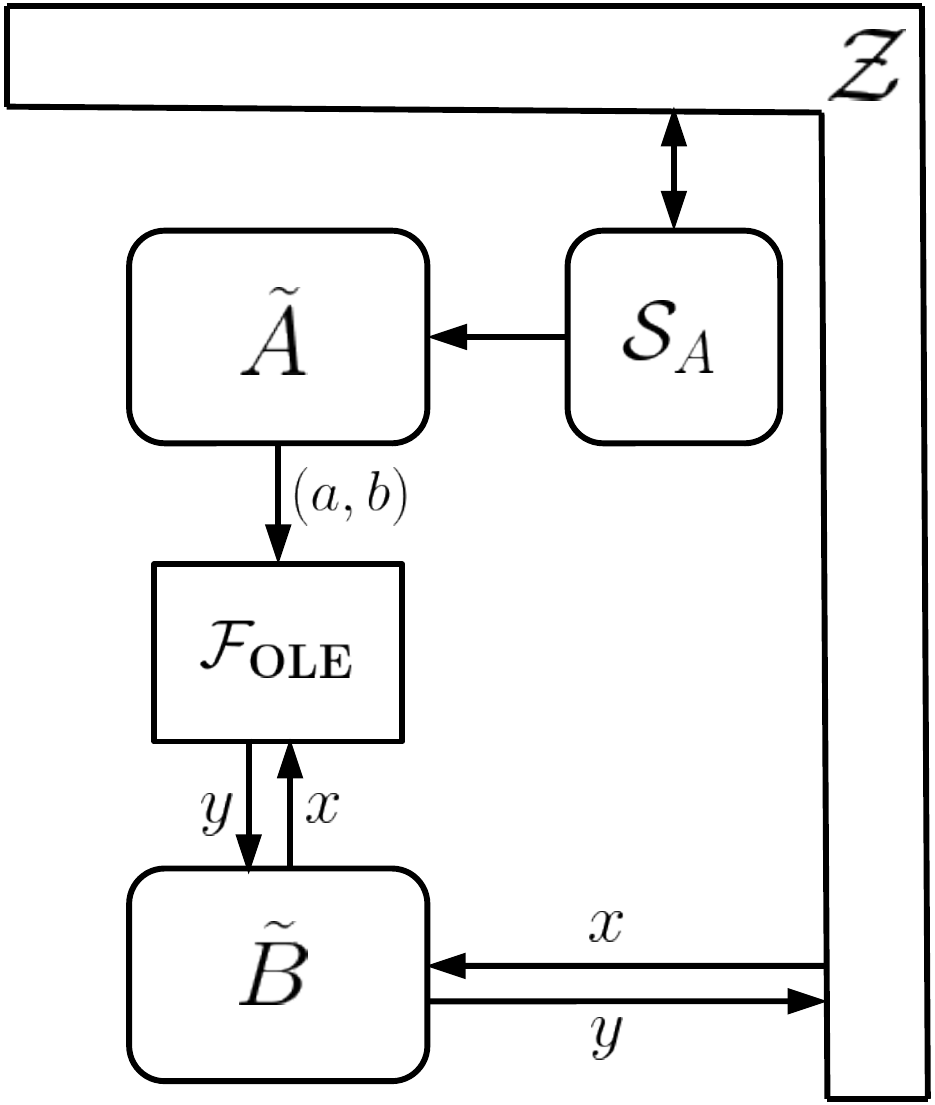}
	\end{minipage}
	\caption{The real and ideal networks in the case of dishonest Alice. In the real network,  Bob honestly follows the protocol $\pi_{\textbf{QOLE}}$, while dishonest Alice $A'$, can be thought as part of the environment $\mathcal{Z}$, since she is controlled by it. In the ideal network, $\tilde{A}$ and $\tilde{B}$ are the dummy parties for Alice and Bob, respectively.}
	\label{fig:DANetworks}
\end{figure}
	To prove statistical quantum-UC security according to Definition \ref{def:statisticalquc}, we first consider a sequence of hybrid protocols from $\mathsf{H}_0$ to $\mathsf{H}_4$. The first hybrid protocol, $\mathsf{H}_0$, is the real execution of the protocol $\mathcal{\pi}_{\textbf{QOLE}}$, and we gradually change it until obtaining the hybrid $\mathsf{H}_4$ which corresponds to the description of the simulator $\mathcal{S}_A$. By proving indistinguishaility of the hybrids throughout the sequence, we show statistical quantum-UC security for $\mathcal{\pi}_{\textbf{QOLE}}$ in the case of dishonest Alice.
	Let us know present the hybrids in detail:
	
	\
	
	\textbf{Hybrid $\mathsf{H}_0$:} This is the real execution of the protocol $\mathcal{\pi}_{\textbf{QOLE}}$.

\

	\textbf{Hybrid $\mathsf{H}_1$:} This hybrid is identical to  $\mathsf{H}_0$, except that we replace the $\mathcal{F}_{\textbf{COM}}$ with a fake commitment functionality, $\mathcal{F}_{\textbf{FakeCOM}}$, in which Bob, i.e. the honest party, can commit no value. This fake functionality works as follows: 
		\begin{itemize}
		\item Commitment phase: expects a \texttt{commit} message from Bob instead of (\texttt{commit, $x$}).
		\item Open phase: expects a message (\texttt{open}, $x$) (instead of open) and sends (\texttt{open}, $x$) to Alice.
	\end{itemize}
	Hybrids $\mathsf{H}_0$ and $ \mathsf{H}_1$ are perfectly indistinguishable, as the simulator still opens the commitments in the same way.
	
	\
	
		\textbf{Hybrid $\mathsf{H}_2$:} This hybrid is identical to $ \mathsf{H}_1$, except that now $\mathcal{S}_A$ prepares entangled states $\ket{B_{0,0}}_{Q_A Q_S}$ instead of $\ket{e^{x^0_i}_{r_i}}_{ i\in [m]}$, and sends the subsystem $Q_A$ to Alice. Additionally, upon receiving  the set of indices, $T$, from Alice, $\mathcal{S}_A$  measures the corresponding elements of subsystem $Q_S$ using $tn$ randomly chosen bases $x^0_i$ and provides $(\texttt{open}, (i, x^0_i, r_i))$ to  $\mathcal{F}_{\textbf{FakeCOM}},\ \forall i\in T$.

	\begin{claim}
		The hybrids $\mathsf{H}_1$ and $\mathsf{H}_2$ are indistinguishable.
		\label{claim:h1h2Alice}
	\end{claim}
	\begin{proof}
		From Alice's point of view, the state received  is exactly the same in both hybrids. In $\mathsf{H}_1$, since the elements $r$  are chosen randomly,
		\begin{align}
		\frac{1}{d}\sum_{r=0}^{d-1} \ketbra{e^{x^0}_{r}} = \frac{\mathds{1}_A}{d},
		\end{align}
		for each $x^0 = 0, \ldots, d-1$. In $\mathsf{H}_2$
		\begin{align}
		\Tr_{Q_S}{\ketbra{B_{0,0}}} = \frac{\mathds{1}_A}{d}.
		\end{align}
		Thus, the environment is not able to distinguish the two scenarios. Furthermore, upon Alice's request of the test set, $T$, the simulator  measures in random bases, $x^0_i$  for $i\in T$, the corresponding qudits of subsystem $Q_S$. Since both entangled qudits in $Q_A Q_S$ get projected to the some random state, $r_i$  for $i\in T$,  $\mathcal{F}_{\textbf{FakeCOM}}$  provides the correct pair $(x^0_i, r_i)_{ i\in T}$ to Alice. Hence, the hybrids $\mathsf{H}_1$ and $\mathsf{H}_2$ are indistinguishable.
	\end{proof}
	
	\
	
	\textbf{Hybrid $\mathsf{H}_3$:} This hybrid is identical to $\mathsf{H}_2$, except that now $\mathcal{S}_A$ extracts Alice's elements $\mathbf{F}_0 = (\bm{a}_0, \bm{b}_0)$  by applying a joint measurement on the systems $\hat{Q}_A Q_S$ in the generalized Bell basis.

	Hybrids $\mathsf{H}_2$ and $ \mathsf{H}_3$ are perfectly indistinguishable, as the simulator only changes the measurement basis for the received state and does not communicate with Alice.
	
	\
	
	\textbf{Hybrid $\mathsf{H}_4$:} This hybrid is identical to $\mathsf{H}_3$, except that now $\mathcal{S}_A$ generates $\bm{c}'$ uniformly at random. Additionally, upon receiving $\bm{d}$, $\bm{s}$ and $\bm{\kappa}$, the simulator extracts Alice's vectors $(\bm{a}, \bm{b})$ and inputs $(a, b)$ by computing expressions~(\ref{eqn:extract_1}) and (\ref{eqn:extract_2}). Finally, $\mathcal{S}_A$ sends $(a, b)$ to the ideal functionality $\mathcal{F}_{\textbf{OLE}}$. Hybrid $\mathsf{H}_4$ corresponds to the description of the simulator $\mathcal{S}_A$.

	Hybrids $\mathsf{H}_3$ and $ \mathsf{H}_4$ are perfectly indistinguishable for the following reasons: first, from the proof of Claim~\ref{claim:h1h2Alice}, we have that the vector $\bm{x}_0$ looks uniformly random to Alice, and consequently, so does $\bm{c}$. Second, the extraction operations do not require any interaction with Alice.
\end{proof}

We now proceed to the case  where Alice is honest and Bob is dishonest.
\begin{lemma}
	The protocol $\mathcal{\pi}_{\textbf{QOLE}}$ (Figure~\ref{fig:fullprotocol}) statistically quantum-UC realizes (Definition~\ref{def:statisticalquc}) $\mathcal{F}_{\textbf{OLE}}$ in the $\mathcal{F}_{\textbf{COM}}-$hybrid model in the case of honest Alice and dishonest Bob.
	\label{lemma:dishonestBob}
\end{lemma}
\begin{proof}
	We start by presenting the simulator $\mathcal{S}_B$ for the case where Bob is dishonest in Figure~\ref{fig:simulator_dis_Bob}. Moreover, in Figure~\ref{fig:DBNetworks}, we illustrate the corresponding real and ideal networks.
	\begin{figure}[h!]
		\centering
		\framebox[\linewidth][l]{%
			\parbox{0.95\linewidth}{%
				\begin{center}
					\textbf{Simulator $\mathcal{S}_{\text{B}}$}
				\end{center}
				
				\textit{(Quantum phase:)}
				
				\begin{enumerate}
					\item $\mathcal{S}_B$ receives the qudits from Bob and tests them as in the protocol $\mathcal{\pi}_{\textbf{QOLE}}$.
					\item $\mathcal{S}_B$ randomly chooses vectors $\bm{a}_0$ and $\bm{b}_0$ and applies  $V^{b^0_i}_{a^0_i}$, $i\in [n]$ to the received qudits.
					\item $\mathcal{S}_B$ extracts the input element $\bm{x}_0$  from  $\mathcal{F}_{\textbf{COM}}$.
				\end{enumerate}
				
				\textit{(Post-processing phase:)}
				
				\begin{enumerate}
					\setcounter{enumi}{3} 
					\item Upon receiving $\bm{c}$ from Bob, $\mathcal{S}_B$ extracts his input $x$ as  $\bm{x}=\bm{c} + \bm{x}_0$.
					\item $\mathcal{S}_B$ sends $x$ to  $\mathcal{F}_{\textbf{OLE}}$ and receives $y$.
					\item $\mathcal{S}_B$ randomly generates the elements $a'\leftarrow_{\$} \mathbb{Z}_d$, $\bm{\kappa}\leftarrow_{\$} \mathbb{Z}^n_d$ and $a_2, \ldots, a_n, b_2, \ldots, b_n\leftarrow_{\$}\mathbb{Z}_d$. 
					\item $\mathcal{S}_B$ computes $b' = a'x - y$, $a_1 = \big(a' - \sum_{i=2}^{n} a_i \kappa_i\big)/\kappa_1$ and $b_1 = \big(b' - \sum_{i=2}^{n} b_i \kappa_i\big)/\kappa_1$.
					\item  $\mathcal{S}_B$ sends $\bm{d} = \bm{a} - \bm{a}_0$, $\bm{s} = \bm{b}_0 + \bm{a} * \bm{c} + \bm{b}$ and $\bm{\kappa}$ to Bob.

				\end{enumerate} 
			}%
		}
		\caption{Simulator $\mathcal{S}_B$ against dishonest Bob.}
		\label{fig:simulator_dis_Bob}
	\end{figure}

	Then, just like in the previous case for dishonest Alice, we consider the respective sequence of hybrid protocols, from $\mathsf{H}_0$ (execution of the protocol) to $\mathsf{H}_2$ (simulator $\mathcal{S}_B$), and prove that they are indistinguishable in the case of dishonest Bob.
	
\textbf{Hybrid $\mathsf{H}_0$:} This is the execution of the real protocol $\mathcal{\pi}_{\textbf{QOLE}}$. In this hybrid, $\mathcal{S}_B$ behaves just like honest Alice up to step 6 of $\mathcal{\pi}_{\textbf{QOLE}}$:  tests the received qudits  (steps 1-4), randomly generates $n$ pairs $(a^0_i, b^0_i)_{ i\in [n]}$ (step 5), and applies the respective operators $V^{b^0_i}_{a^0_i}$ ${ \text{ for } i\in [n]}$ to the received states (step 6).
	\begin{figure}[h!]
		\centering
		\begin{minipage}[l]{0.5\linewidth}
			\includegraphics[scale=0.3]{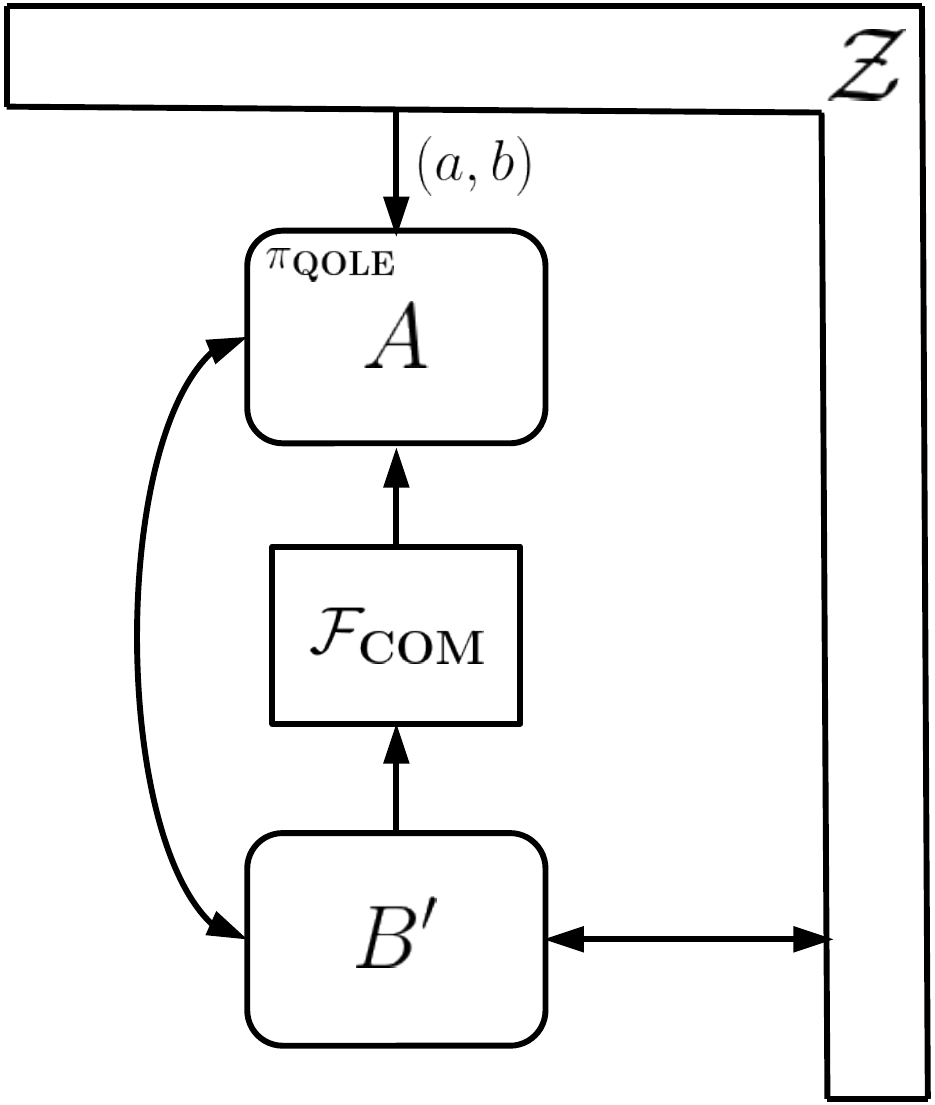}\end{minipage}\begin{minipage}[r]{0.5\linewidth}
			\includegraphics[scale=0.3]{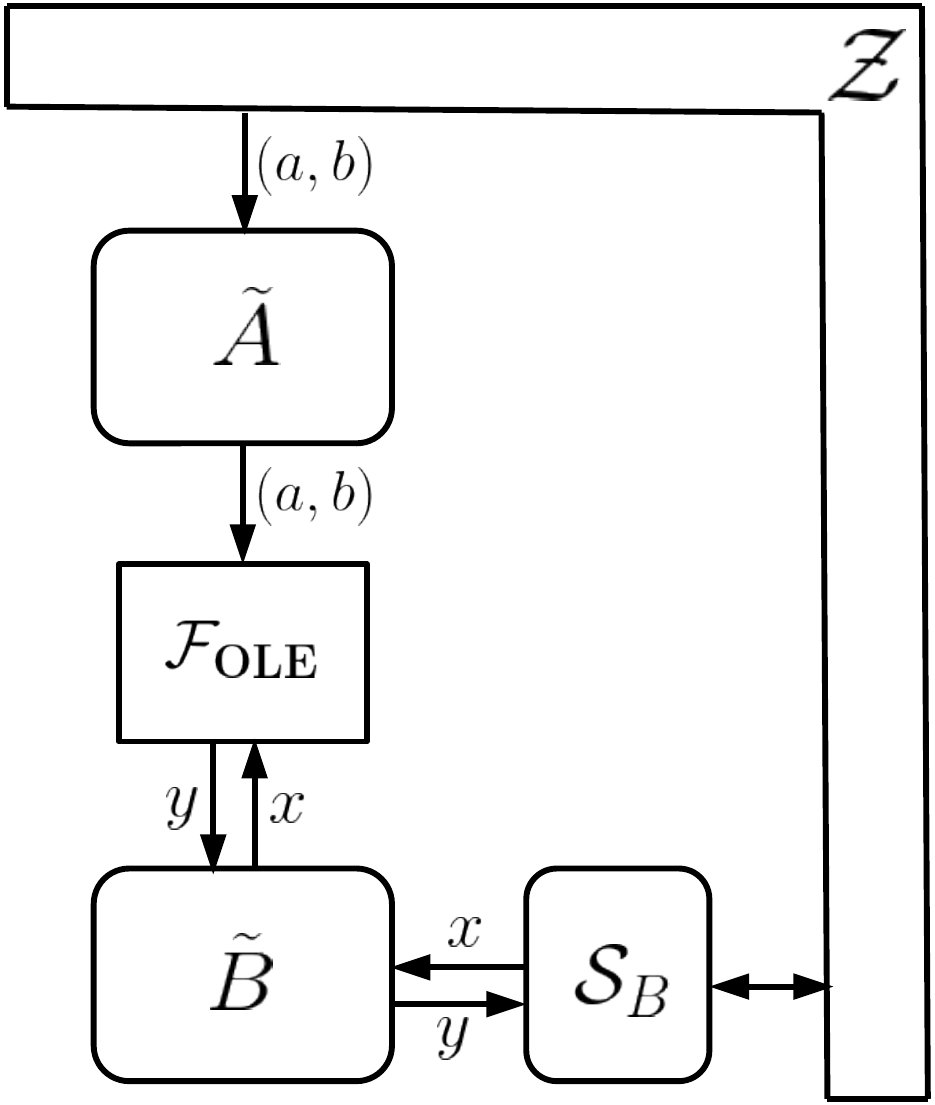}
		\end{minipage}
		\caption{{The real and ideal networks in the case of dishonest Bob. In the real network,  Alice honestly follows the protocol $\pi_{\textbf{QOLE}}$, while dishonest Bob $B'$, can be thought as part the environment $\mathcal{Z}$, since he is controlled by it. In the ideal network, $\tilde{A}$ and $\tilde{B}$ are the dummy parties for Alice and Bob, respectively.} }
		\label{fig:DBNetworks}
	\end{figure}

		\textbf{Hybrid $\mathsf{H}_1$:} This hybrid is identical to $\mathsf{H}_0$, except that now $\mathcal{S}_B$ extracts Bob's random vector $\bm{x}_0$ from the commitment functionality $\mathcal{F}_{\textbf{COM}}$. Additionally, upon receiving $\bm{c}$ from Bob, $\mathcal{S}_B$ extracts Bob's input $x$ by computing $\bm{c} + \bm{x}_0$. Then,  $\mathcal{S}_B$ sends the extracted element $x$ to  $\mathcal{F}_{\textbf{OLE}}$ and receives $y$.

	Hybrids $\mathsf{H}_0$ and $\mathsf{H}_1$ are perfectly indistinguishable, because $\mathcal{S}_B$ only interacts with Bob when receiving the element $\bm{c}$, and this does not change anything from Bob's point of view. The corresponding operations are either carried out locally by $\mathcal{S}_B$ or along with  $\mathcal{F}_{\textbf{COM}}$ which, by definition, is fully controlled by $\mathcal{S}_B$.
	
	\
	
	\textbf{Hybrid $\mathsf{H}_2$:} This hybrid is identical to $\mathsf{H}_1$,  except that now $\mathcal{S_B}$ generates $(a, b)$, $\bm{d}$ and $\bm{s}$: it starts by randomly generating $a'\leftarrow_{\$} \mathbb{Z}_d$, $\bm{\kappa}\leftarrow_{\$} \mathbb{Z}^n_d$ and $a_2, \ldots, a_n, b_2, \ldots, b_n\leftarrow_{\$}\mathbb{Z}_d$. Then, it computes  $b'$  according to the generated $a'$, the extracted element $x$ and the output $y$ of $\mathcal{F}_{\textbf{OLE}}$, as $b' = a'x - y$. It  masks $a'$ and $b'$ as         \begin{align}
	a' = \bm{a} \cdot \bm{\kappa} \ \ \ \ \ \text{ and }\ \ \ \ \ 
	b' = \bm{b} \cdot \bm{\kappa}, 
	\end{align}
	by setting $a_1$ and $b_1$ accordingly, i.e. $a_1 = \big(a' - \sum_{i=2}^{n} a_i \kappa_i\big)/\kappa_1$ and $b_1 = \big(b' - \sum_{i=2}^{n} b_i \kappa_i\big)/\kappa_1$. Finally, $\mathcal{S}_B$ sends $\bm{d} = \bm{a} - \bm{a}_0$, $\bm{s} = \bm{b}_0 + \bm{a}  * \bm{c} + \bm{b}$ and $\bm{\kappa}$ to Bob. This is the last hybrid of the sequence and corresponds to the description of $\mathcal{S}_B$.

	\begin{claim}
		The hybrids $\mathsf{H}_1$ and $\mathsf{H}_2$ are indistinguishable.
		\label{claim:h1h2Bob}
	\end{claim}
	\begin{proof}
			Since, in its first two steps, $\mathcal{S}_B$ executes a RWOLE scheme, according to Lemma~\ref{lemma:wrole_dishonest_bob} we have that $\mathcal{S}_B$ is $\epsilon-$close to a situation where Bob's knowledge on the vectors $(\bm{a}_0, \bm{b}_0)$ is lower-bounded by the value
			\begin{align}
		\frac{1}{n}\lambda(\zeta) = \frac{\log d}{2}(1-h_d(\zeta)),
		\end{align}
		for $\zeta\in\,]0, 1-\frac{1}{d}]$, $n$ the security parameter and $\epsilon(\zeta, n) = \exp( -\frac{2 \zeta^2t^2n^2}{(nt+1)(t+1)})$. Also, as Bob receives $\bm{d}$ and $\bm{s}$, according to Lemma~\ref{lemma:wole_bob_dishonest}, his knowledge on $(\bm{a}, \bm{b})$ is also lower-bounded by the same $\lambda(\zeta)/n$. Furthermore, since $\mathcal{S}_B$ defines $\bm{a}$ such that $a' = \bm{a} \cdot \bm{\kappa}$, from Lemma~\ref{lemma:extraction} we can conclude that $\mathcal{S}_B$ is $(\xi + \epsilon)-$close to a scenario where $a'$ is uniformly distributed. This comes from the properties in Lemma~\ref{lemma:extraction} and the triangle inequality:
		\begin{align}
			\delta( \tau_{\mathbb{Z}_d} \otimes \sigma_{B'},\, \rho_{F_a B'} )  &\leq\delta( \tau_{\mathbb{Z}_d} \otimes \sigma_{B'},\, \sigma_{F_{a} B'} ) + \delta( \sigma_{F_{a} B'},\, \rho_{F_a B'} )\nonumber \\
			& \leq K\, 2^{-n \, f_d(\zeta)} + e^{ -\frac{2 \zeta^2t^2n^2}{(nt+1)(t+1)}}= \xi + \epsilon,
		\end{align}
		where $K = \frac{\sqrt{d}}{2}$, $f_d(\zeta) = \frac{\log d}{4} (1-h_d(\zeta))$. This means that the triple $(\bm{d}, \bm{s}, \bm{\kappa})$ only gives to the environment a negligible advantage in distinguishing between the real and ideal world executions.
	\end{proof}
\end{proof}

\section{Protocol Generalizations}
\label{sec:protgeneral}

\subsection{QOLE in Galois fields  of prime-power dimensions} \label{subsec:galois_ext}

So far, we have been working in Hilbert spaces of prime dimensions; this reflects the fact that, for prime $d$,  $\mathbb{Z}_d$ is a field  and, under a well-defined set of MUBs $\{ \ket{e^x_r} \}_{r\in\mathbb{Z}_d},\ \forall x\in\mathbb{Z}_d$, we have the affine relation \eqref{eq:main_ingredient}:
\begin{align}
V^b_a \ket{e^x_r} = c_{a, b, x, r} \ket{e^x_{ax - b +r}}.
\end{align}

In this section, we generalize our protocol, $\pi_{\textbf{QOLE}}$, to Hilbert spaces of prime-power dimensions, $N=d^M$ ($d$ prime and $M>1$), taking advantage  of the fact that in a Galois field of dimension $d^M$, $GF(d^M)$, we can build a complete set of $N + 1$ MUBs \cite{DEBZ10}.

Succinctly, in $GF(d^M)$, we identify the integers $i\in\mathbb{Z}_N$ with their $d-$ary representation, i.e.
\begin{align}
\mathbb{Z}_N \ni i = \sum^{M-1}_{n=0} i_n d^n \, \longleftrightarrow \, (i_0, \dots , i_{M-1}) \in GF(d^M).
\end{align}
In these fields there are  two operations, addition and multiplication, which we denote by $\oplus$ and $\odot$, respectively. Addition  is straightforward, as it is given by the component-wise addition modulo $d$ of elements, i.e. $i\oplus j = (i_0 + j_0 \mod d, \dots , i_{M-1} + j_{M-1} \mod d)$.  Considering $i = \sum^{M-1}_{n=0} i_n d^n $ as a polynomial of degree $M-1$ given by $i(p) = \sum^{M-1}_{n=0} i_n p^n$,  multiplication between two elements $i, j$, is given by the multiplication between the corresponding polynomials $i(p)$ and $j(p)$ modulo some irreducible polynomial $m(p)$, i.e. $i\odot j = \big(i(p)\times j(p)\big) \mod m(p)$. 

Analogously to prime-dimension fields, we can write the operators $V_a^b$ in the computational basis, as  \begin{align}
V^b_a = \sum_{k=0}^{N-1} \ket{k \oplus a} \omega^{(k \oplus a)\odot b} \bra{k},
\end{align}
and the eigenstates for the corresponding $N+1$ pairwise MUBs, as 
\begin{align}
\ket{e^x_r} = \frac{1}{\sqrt{N}} \sum_{l=0}^{N-1}\ket{l} \omega^{\ominus(r \odot l)} \alpha^{x*}_{\ominus l},
\end{align}
where $\alpha^{x}_{\ominus l}$ is a phase factor whose form depends on whether $d$ is even or odd. For details, see Section 2.4.2 in \cite{DEBZ10}.

Given the above, we can derive (Appendix \ref{app:generalrelation}) the following affine relation similar to \eqref{eq:main_ingredient}:
\begin{align}
V^b_a\ket{e^i_r} = \omega^{r\odot a} \alpha^{i*}_a \ket{e^i_{i\odot a\ominus b\oplus r}}.
\label{eq:general_relation}
\end{align}
Notice that all the steps in the $\pi_{\textbf{QOLE}}$ depend on the properties of the field operations and on the fact that  \eqref{eq:main_ingredient} holds. Hence, we can use  $\pi_{\textbf{QOLE}}$ adapted for the operations $\oplus$ and $\odot$, in order to quantum-UC-realize $\mathcal{F}_{\textbf{OLE}}$ in  fields of prime-power dimension $d^M$.

\subsection{Quantum Vector OLE}
\label{subsec:qvole}
In the proposed protocol $\pi_{\textbf{QOLE}}$, we extract one instance of OLE out of $n$ instances of WOLE. As far as efficiency is concerned, it would be desirable to generate more instances of OLE out of those $n$ instances of WOLE. Here, we show how to use WOLE as a resource to realize the VOLE functionality, $\mathcal{F}_{\textbf{VOLE}}^k$, presented in Figure \ref{fig:func_vole}. In this case, Alice  fixes a $k$ (which is specified later), defines a set of $k$ linear functions $(\bm{a}, \bm{b})\in\mathbb{F}^k_q\times\mathbb{F}^k_q$ and Bob outputs the evaluation of all these functions on a specified element $x\in\mathbb{F}_q$ that he chooses, i.e. $\bm{f}:=\bm{a} x+ \bm{b}$. Since  $\pi_{\textbf{QOLE}}$ can be extended to finite fields $\mathbb{F}_q$, where $q$ is a prime or prime-power number (as shown above), the $\mathcal{F}_{\textbf{VOLE}}^k$ functionality can also be defined in $\mathbb{F}_q$. 

In the extraction phase of $\pi_{\textbf{QOLE}}$, Alice randomly chooses a function $g_{\bm{\kappa}}$ and applies it to the pair $(\bm{a},\bm{b})$. This procedure suggests that, in order to generate different input elements $(a', b')$, Alice can randomly choose another function $g_{\bm{\kappa'}}$ and set $a' = g_{\bm{\kappa'}}(\bm{a})$ and $b' = g_{\bm{\kappa'}}(\bm{b})$. This is equivalent to generating a random $2\times n$ matrix in  $\mathbb{F}_q$, i.e.
\begin{align}
\left[
\begin{array}{ccc}
\horzbar & \bm{\kappa} & \horzbar \\
\horzbar & \bm{\kappa'} & \horzbar \\
\end{array}
\right] \left[
\begin{array}{cc}
\vertbar & \vertbar  \\
\bm{a} & \bm{b} \\
\vertbar  & \vertbar \\
\end{array}
\right]  = \left[
\begin{array}{cc}
a & b \\
a' & b'
\end{array}
\right].
\end{align}
However, in case $\bm{\kappa}$ and $\bm{\kappa'}$ are linearly dependent (i.e. $\bm{\kappa} = c\bm{\kappa'}$ for some $c\in\mathbb{Z}_d$), Bob would have some extra information about Alice's elements $(a,b)$ and $(a', b')$, as $(a,b) = c( a', b')$. This leads to a situation beyond the  $\mathcal{F}_{\textbf{VOLE}}^k$ definition. 
To avoid this issue, let us consider the set of $k\times n$ matrices with rank $k$ over $\mathbb{F}_q$ for $1\leq k\leq n$, and denote it by $\mathcal{R}_{k\times n}(\mathbb{F}_q)$. For a binary finite field,  $\mathcal{R}_{k\times n}(\mathbb{F}_2)$ is a two-universal hash family \cite{D21, CW79}. Similarly, one can prove that the more general set $\mathcal{R}_{k\times n}(\mathbb{F}_q)$ is also a two-universal hash family from $\mathbb{F}_q^n$ to $\mathbb{F}_q^k$. During the extraction phase of the original $\pi_{\textbf{QOLE}}$, Alice chooses vectors $(\bm{a}, \bm{b})$ according to the random vector $\bm{\kappa}$ and the desired final elements $(a,b)$ (see step 9 in Figure \ref{fig:fullprotocol}). In that case, since there is only one random vector $\bm{\kappa}$, there are $n-1$ undefined variables for each vector $\bm{a}$ and $\bm{b}$, i.e. $a_2, \ldots, a_n$ and $b_2, \ldots, b_n$ that can be chosen freely. For the VOLE protocol, instead of choosing just one vector $\bm{\kappa}$, Alice randomly chooses a  matrix $\mathcal{K}\in\mathcal{R}_{k\times n}(\mathbb{F}_q)$ of rank $k$. She then defines vectors $(\bm{a}', \bm{b}')\in \mathbb{F}^n_q\times \mathbb{F}^n_q$ consistent with the final elements $(\bm{a}, \bm{b})\in \mathbb{F}^k_q\times \mathbb{F}^k_q$. That is, Alice has the following system:
\begin{align}
	&\left[
	\begin{array}{ccc}
		\bighorzbar & \bm{\kappa}_1 & \bighorzbar \\
		& \ldots & \\
		\bighorzbar & \bm{\kappa}_k & \bighorzbar \\
	\end{array}
	\right] \left[
	\begin{array}{cc}
		\bigvertbar & \bigvertbar  \\
		\bm{a'} & \bm{b'} \\
		\bigvertbar  & \bigvertbar \\
	\end{array}
	\right]  = \left[
	\begin{array}{cc}
		\vertbar & \vertbar\\
		\bm{a} & \bm{b}\\
		\vertbar & \vertbar\\
	\end{array}
	\right],
\end{align}
that can be solved by means of the Gaussian elimination method. Since $\mathcal{K}\in \mathcal{R}_{k\times n}(\mathbb{F}_q)$, there will be $n-k$ undefined variables in both vectors $\bm{a}'$ and $\bm{b}'$. Let $U$ denote the set of undefined indexes in $\bm{a}'$ and $\bm{b}'$. Alice randomly chooses $a'_i$ and $b'_i$ for $i\in U$ and solves the above equation system. Then, they proceed similarly to the original $\pi_{\textbf{QOLE}}$ and execute the derandomization protocol $\pi^n_{\textbf{WOLE}}((\bm{a}',\bm{b}'), \bm{x})$. Finally, Bob applies Alice's chosen matrix $\mathcal{K}$ to his output vector $\bm{y}'$ to get the final element $\bm{y}$. This vectorized extraction protocol $\pi_{\textbf{VEXT}}$ is presented in Figure~\ref{fig:vext}. 
\begin{figure}[h!]
	\centering
	\framebox[\linewidth][l]{%
		\parbox{0.95\linewidth}{%
			\begin{center}
				\textbf{Protocol $\mathcal{\pi}_{\textbf{VEXT}}$}
			\end{center}
			
			\textbf{Alice's input:} $(\bm{a}, \bm{b}) \in \mathbb{F}^k_q\times \mathbb{F}^k_q$;	\textbf{Bob's input:} $x\in \mathbb{F}_q$
			
			\begin{enumerate}
				\item Alice chooses randomly a matrix $\mathcal{K} \in \mathcal{R}_{k\times n}(\mathbb{F}_q)$ and sends it to Bob.
				\item Using the Gaussian elimination method, Alice finds one solution of the system:
								\begin{align}
					\mathcal{K} \left[
					\begin{array}{cc}
						\bigvertbar & \bigvertbar  \\
						\bm{a'} & \bm{b'} \\
						\bigvertbar  & \bigvertbar \\
					\end{array}
					\right]  = \left[
					\begin{array}{cc}
						\vertbar & \vertbar\\
						\bm{a} & \bm{b}\\
						\vertbar & \vertbar\\
					\end{array}
					\right]
					\end{align}
				
				\begin{enumerate}
					\item Alice finds the set $U$ of undefined indexes in $\bm{a}'$ and $\bm{b}'$.
					\item Alice randomly generates $a'_i, b'_i \leftarrow_{\$}\mathbb{F}_q$ for $i\in U$.
					\item Alice solves the system for indexes $i\notin U$.
				\end{enumerate}
				\item Alice and Bob run $\mathcal{\pi}^n_{\textbf{WOLE}}((\bm{a}', \bm{b}'), \bm{x})$, where $\bm{x} = (x, \ldots, x)$. Bob outputs  $\bm{y}'\in \mathbb{F}_q^n$.
				
				\item Bob computes $\bm{y} = \mathcal{K}\bm{y}'$.
			\end{enumerate}
			
			\textbf{Alice's output:} $\bot$; \textbf{Bob's output:} $\bm{y}\in \mathbb{F}^k_q$
			}
	}
	\caption{Extraction protocol  for VOLE.}
	\label{fig:vext}
\end{figure}

The correctness of the protocol is drawn immediately from linearity:
\begin{align}
	\bm{y} = \mathcal{K} \bm{y}'
	= \mathcal{K}(\bm{a}'x + \bm{b}')
	= \bm{a}x + \bm{b}.
\end{align}
The security of the protocol is constrained by  $\xi = \frac{1}{2}\sqrt{2^{k\log q - H_\text{min}(X|E)}}$ given by Lemma~\ref{lem:leftover}, where we consider $l$ to be $k$ and $d$ to be $q$. As before, $\mathbf{F}_{\bm{a}'}$ denotes the distribution of the $\pi^n_{\textbf{WOLE}}$ protocol's input ${\bm{a}'}$ from Bob's perspective. From Lemma~\ref{lem:leftover}, since $\mathcal{R}_{k\times n}(\mathbb{F}_q)$ is a two-universal family of hash functions, we know that $\mathcal{K}\in \mathcal{R}_{k\times n}(\mathbb{F}_q)$ approximates $\mathcal{K}\mathbf{F}_{\bm{a}'} = \mathbf{F}_{\bm{a}}$ to uniform conditioned on Bob's side information. However, the closeness parameter $\xi$, has to be negligible in the security parameter $n$, thus setting a bound on $k$ (the size of VOLE), i.e. for $\eta > 0$,
\begin{align}
	k \log q - \frac{n \log q}{2}\big(1-h_{q}(\zeta)\big) & < -n \eta \log q\\
	k& < n \Big(\frac{1}{2}\big(1-h_q(\zeta)\big) - \eta\Big).
\end{align}
Since $k>0$, we have that  $ 0 < \eta < \frac{1}{2}\big(1-h_q(\zeta)\big)$. This gives a bound on the elements that we can extract from $n$ WOLEs  and  shows how Alice can fix $k$ in the beginning.  Note that this bound is not necessarily optimal, and one could try to improve it. We leave this as future work, as it goes beyond the scope of this paper.

Let us denote by $\mathcal{\pi}_{\textbf{QVOLE}}$ the protocol $\mathcal{\pi}_{\textbf{QOLE}}$ with the subprotocol $\mathcal{\pi}_{\textbf{VEXT}}$ instead of $\mathcal{\pi}_{\textbf{EXT}}$. For the security of $\mathcal{\pi}_{\textbf{QVOLE}}$, we have:

\begin{theorem}[quantum-UC security of $\mathcal{\pi}_{\textbf{QVOLE}}$]
	
	The protocol $\mathcal{\pi}_{\textbf{QOVLE}}$ statistically quantum-UC realizes  (see Definition~\ref{def:statisticalquc}) $\mathcal{F}_{\textbf{VOLE}}^k$ in the $\mathcal{F}_{\textbf{COM}}-$hybrid model.
	\label{thm:QUC-VOLE}
\end{theorem}

The proof is much the same as the proof of Theorem~\ref{thm:QUC}, therefore we omit it.

\section{Outlook and further work} 
\label{sec:outlook}

OLE is an important primitive for secure two-party computation, and while for other primitives such as bit commitment, OT  and coin flipping there is a plethora of both theoretical as well as  concrete protocol proposals \cite{M05, MTVUZ05, BBBG09, BBBGST11, KWW12, NJCKW12, KC13, PJLCLTKD14, LAAPMP14, LAPPP16, ARW19, BCKD20,  ARV21,  SMP22}, up until now, there was no OLE protocol based on quantum communication. In this work, we present two protocols for QOLE. The first protocol is secure against semi-honest adversaries in the static corruption setting. The second proposed protocol, $\pi_{\textbf{QOLE}}$, builds upon the semi-honest version and extends it to the dishonest case, following a commit-and-open approach. We prove this second protocol to be secure in the quantum-UC framework when assuming ideal commitments, making it possible to be composed in any arbitrary way. We also constructed two generalizations of our protocol: the first achieves QOLE in Galois fields of prime-power dimensions and the second is a protocol for quantum vector OLE.
Note that our protocol achieves everlasting security, i.e. it remains information-theoretically secure after its execution, even if the dishonest party becomes more powerful in the future.

While the security of $\pi_{\textbf{QOLE}}$ is thoroughly analyzed, this is done in the noiseless case. Proving security assuming the existence of noise should follow a similar reasoning (see also \cite{DFLSS09}). Indeed, in the presence of noise, in Step 4 in the quantum phase of $\pi_{\textbf{QOLE}}$ (Figure~\ref{fig:fullprotocol}), Alice should abort the protocol if the errors measured ($err$) exceed some predetermined value $\nu$, that is assumed to be due to noise. This way, the error parameter is $err = \nu + \zeta '$, where $\zeta '$ accounts for the activity of a dishonest Bob. Naturally, this will decrease the bound on the min-entropy of Alice's functions $\textbf{F}$ given Bob's side information, i.e.
 \begin{align}
 H_{\min}(\mathbf{F} | \mathbf{Y} E)_{\sigma_{\mathbf{F}\mathbf{Y} E}} \geq \frac{n\log d}{2}\left(1 - h_d(\nu + \zeta ' + \zeta)\right).
 \end{align}
Although it is possible to generalize the security results to the setting of noisy quantum communication, it is not guaranteed that the proposed protocol retains correctness. Therefore, as future work, it would be useful to propose  specific implementations where noise is taken into account and its effect on the correctness and security is studied. Furthermore, the protocol's resilience to practical attacks \cite{BCDP21} can be investigated.

Following a different assumption model, the $\pi_{\textbf{QOLE}}$ protocol can be easily adapted to remain secure in the bounded-quantum-storage model. In this adapted version, the test phase of $\pi^n_{\textbf{RWOLE}}$ is simply substituted by a waiting time $\Delta t$. This ensures that Bob is only able to store a noisy or limited amount of qudits. It would be interesting to explore how different noisy channels affect the security properties. Also, to guarantee the composability of the protocol, an analysis in the bounded-quantum-storage-UC model as put forth by Unruh \cite{U11} can be performed.

Our protocol is a two-way protocol, i.e. Bob prepares and sends a quantum state, Alice applies some operation to it, and sends it back to Bob who measures the final state. For QOT there exist several proposals for two-way  protocols \cite{ASRP21, KST21, CKS10, CGS16}, and, in particular, the one presented by Amiri et al. \cite{ASRP21} also demonstrates their experimental feasibility. This is further motivation to work on developing realistic practical implementations of our protocol. 
Furthermore, one could increase the security standards  by making our protocol device-independent. The proposal of Kundu et al. \cite{KST21}, who extended the work of Chailloux et al. \cite{CKS10}, could serve as an inspiration. And while the aforementioned works focus on two-way QOT protocols,  recently, non-interactive  or one-way protocols have also been proposed for device-independent \cite{DIQOT22} and XOR QOT \cite{XORQOT22}.

Finally, based on our results, one could construct quantum protocols for oblivious polynomial evaluation, which -- as mentioned in Section \ref{sec:Intro} -- is another important primitive facilitating various applications.
\section*{Acknowlegments}

This work was funded by Fundação para a Ciência e a Tecnologia (FCT) through National Funds under Award SFRH/BD/144806/2019, Award UIDB/50008/2020, and Award UIDP/50008/2020; in part by the Regional Operational Program of Lisbon; by FCT, I.P. and, when eligible, by COMPETE 2020 FEDER funds, through the Competitiveness and Internationalization Operational Programme (COMPETE 2020), under the project QuantumPrime PTDC/EEI-TEL/8017/2020, and under the Scientific Employment Stimulus - Individual Call (CEEC Individual) - 2020.03274.CEECIND/CP1621/CT0003.

\bibliographystyle{quantum}
\bibliography{bibl}
	\appendix
\section{Security of the RWOLE protocol -- Proofs}

\subsection{Proof of Lemma \ref{lemma:wrole_dishonest_bob} (Dishonest Bob)}
\label{app:proofBobdishonest}
As we mentioned in the main text, this proof is  a combination and adaptation to our case of results from \cite{DFLSS09} and \cite{Dupuis2015}. 

To simplify the notation, in this proof we drop the subscript $0$ that refers to the RWOLE phase, e.g. we write Alice's function vector  $\mathbf{F}_0$,  simply  as $\mathbf{F}$.

Let the values that Bob commits be fixed as  $(x_i, r_i)\, \forall i\in[m]$, where $m = (1+t)n$ and $t n$  the number of qudits $\ket{e^{x_i}_{r_i}}$ used in the \textit{Test Phase} to check whether he is honest or not. Throughout the proof we denote by $\boldsymbol{x} = (x_1, \ldots, x_m)$ and $\boldsymbol{r} = (r_1, \ldots, r_m)$ the vectors in $\mathbb{Z}_d^m$ whose components contain Bob's commitments $\forall i\in [m]$, and by $X^m=(X_1,\ldots,X_m)$ the vector of the  random variables associated to  $x_i, i\in [m]$. For each  pair $(x_i, r_i)$, the corresponding qudit $\ket{e^{x_i}_{r_i}}$ belongs in the Hilbert space $\mathcal{H}_{X_i}$, and the quantum system including all the qudits is in $\mathcal{H}_{X^m}=\bigotimes_{i\in [m]}\mathcal{H}_{X_i}$. For simplicity, we refer to the quantum systems in terms of the corresponding  random variables  $X_i$, instead of the Hilbert spaces $\mathcal{H}_{X_i}$.

Recall that the set $T\subset[m]$ contains the $tn$ indices of the test qudits, and by $\bar{T}$ we denote its complement $[m]\backslash T$.  For  $i\in T$, $\boldsymbol{x}_{|T}$ is the vector whose components are the  bases $x_i$  in which Alice will measure the test qudits, and $ \boldsymbol{r'}_{|T}$ is the vector whose components are Alice's measurement results. The corresponding quantum system is in the Hilbert space $\bigotimes_{i\in T}\mathcal{H}_{X_i}$, which for simplicity we denote as $X^m_{|T}$ in terms of the associated random variables.   
Finally, $r_H(\cdot, \cdot)=d_H(\cdot, \cdot)/n$ is the relative Hamming distance between two  vectors of size $n$, with $d_H(\cdot, \cdot)$ being their Hamming distance.

\begin{proof}
	Let us start by proving property 1. of Lemma \ref{lemma:wrole_dishonest_bob}.
	After the first step of $\mathcal{\pi}^n_{\text{RWOLE}}$ (Figure \ref{fig:wrole}),  the generated state is $\rho_{X^m E}$, where $E$ is an auxiliary quantum system that Bob holds. Without loss of generality we assume that $\rho_{X^m E} = \ketbra{\phi_{X^m E}}$,   i.e, it is a pure state\footnote{Otherwise, we purify it and carry the purification system along with $E$.}. If Bob is honest, we have that $\ket{\phi_{X^m E}}=\ket{e^{\boldsymbol{x}}_{\boldsymbol{r}}}\otimes \ket{\psi_{E}}$, i.e.,   Bob's auxiliary quantum system $E$ is not entangled to the states that he sends to Alice.  
	
	The \textit{Test Phase} of the protocol is used to guarantee that the real state is close to an ideal state that satisfies the properties 1. and 2. of Lemma \ref{lemma:wrole_dishonest_bob}. Let $\boldsymbol{r'}$ be the vector whose components are Alice's  outcomes when measuring the state of  $X^m$  in the committed bases $\bm{x}$, and let $\mathcal{T}$ be the random variable associated to the set of indexes $T$  of size $t n$. We can consider the state:
		\begin{align}
	\rho_{\mathcal{T}X^m E} = \rho_\mathcal{T}\otimes \ketbra{\phi_{X^m E}} = \sum_{T} P_\mathcal{T}(T) \ketbra{T}\otimes \ketbra{\phi_{X^m E}},
	\label{eq:real_definition}
	\end{align}
	to be the state resulting from the real execution of the protocol, and prove that it is close to some state, $\sigma_{\mathcal{T}X^m E}$, that fulfils the following property: for any choice of $T$ and for any outcome $\boldsymbol{r'}_{|T}$ when measuring the state of  ${X^m}_{|T}$ in the bases $\boldsymbol{x}_{|T}$, the relative error $r_H(\boldsymbol{r'}_{|T}, \boldsymbol{r}_{|T})$  is an upper bound on the relative error $r_H(\boldsymbol{r'}_{|\bar{T}}, \boldsymbol{r}_{|\bar{T}})$, which one would obtain by measuring the remaining subsystems ${X^m}_{|\bar{T}}$  in the bases $\boldsymbol{x}_{|\bar{T}}$. This state, $\sigma_{\mathcal{T}X^m E}$, can be  written as: 
	\begin{align}
	\sigma_{\mathcal{T}X^m E} = \sum_{T} P_\mathcal{T}(T) \ketbra{T}\otimes \ketbra{\tilde{\phi}_{X^m E}^{T}},
	\label{eq:ideal_definition}
	\end{align}
	where $\forall T$, 
		\begin{align}
	\ket{\tilde{\phi}_{X^m E}^{T}} = \sum_{\boldsymbol{r'}\in \mathcal{B}_{T}} \alpha_{\boldsymbol{r'}}^{T}\ket{e^{\boldsymbol{x}}_{\boldsymbol{r'}}}\otimes \ket{\psi^{\boldsymbol{r'}}_{E}},
	\label{eq:ideal_definition_2}
	\end{align}
	for  $\mathcal{B}_{T} = \{ \boldsymbol{r'}\in \mathbb{Z}_d^{m}: r_H(\boldsymbol{r'}_{|\bar{T}}, \boldsymbol{r}_{|\bar{T}})\leq r_H(\boldsymbol{r'}_{|T}, \boldsymbol{r}_{|T}) + \zeta \}$  for  $\zeta>0$ and arbitrary coefficients $\alpha^T_{\boldsymbol{r'}}$. The state $\ket{\psi^{\boldsymbol{r'}}_{E}}$ is an arbitrary state on $E$ subsystem that possibly depends on $\boldsymbol{r'}$. Then we have:
		\begin{lemma}
		
		Let the quantum states $\rho_{\mathcal{T}X^m E}$ and $\sigma_{\mathcal{T}X^m E}$   be given by \eqref{eq:real_definition} and \eqref{eq:ideal_definition}, respectively. Then, $\forall \zeta >0$ and fixed strings $\boldsymbol{x}, \boldsymbol{r}\in \mathbb{Z}_d^{m}$, we have that
				\begin{align}
		\rho_{\mathcal{T}X^m E} \approx_\epsilon \sigma_{\mathcal{T}X^m E},
		\end{align}
		where $\epsilon(\zeta, n) = \epsilon(\zeta, n) = \exp( -\frac{2 \zeta^2t^2n^2}{(nt+1)(t+1)})$. That is, the real state $\rho_{\mathcal{T}X^m E}$ is exponentially close, with respect to $n$, to the ideal state $\sigma_{\mathcal{T}X^m E}$.
		\label{lemma:dishonest_bob_2}
	\end{lemma}

	\begin{proof}
		This proof is an adaptation of the proof of Lemma 4.3 from \cite{DFLSS09} to our case.
		
		For any $T$, let $\ket{\tilde{\phi}_{X^m E}^T}$ be  the renormalized projection of $\ket{\phi_{X^m E}}$ into the subspace $$\text{Span}\left\{ \ket{e^{\boldsymbol{x}}_{\boldsymbol{r'}}} : \boldsymbol{r'} \in \mathcal{B}_{T} \right\}\otimes \mathcal{H}_{E},$$ and  let $\ket{\tilde{\phi}_{X^m E}^{T^\perp}}$ be the renormalized projection of $\ket{\phi_{X^m E}}$ into its orthogonal complement. 
		
		We can, then, write
		\begin{align}
		\ket{\phi_{X^m E}} = \epsilon_T \ket{\tilde{\phi}_{X^m E}^T} + \epsilon_T^\perp \ket{\tilde{\phi}_{X^m E}^{T^\perp}},
		\end{align}
		with  $\epsilon_T = \braket{\tilde{\phi}_{X^m E}^T}{\phi_{X^m E}}$ and $\epsilon_T^\perp = \braket{\tilde{\phi}_{X^m E}^{T^\perp}}{\phi_{X^m E}}$. By construction, this state satisfies \eqref{eq:ideal_definition_2}.
		Furthermore, we can calculate the distance:
				\begin{align}
		\delta\big( \ketbra{\phi_{X^m E}}, \ketbra{\tilde{\phi}_{X^m E}^T} \big) = \sqrt{1 - |\braket{\tilde{\phi}_{X^m E}^T}{\phi_{X^m E}}|^2} 
		=\sqrt{1 - |\epsilon_T|^2} 
		= |\epsilon_T^\perp|,
		\end{align}
		where,  given $T$,  $|\epsilon_T^\perp|$ is the probability amplitude for getting  outcome $\boldsymbol{r'} \notin \mathcal{B}_{T}$ when measuring the state of $X^m$ in bases $\boldsymbol{x}$.
		We continue to derive an upper bound on the distance between the real and the ideal state: 
		\begin{align}
		\delta\big( \rho_{\mathcal{T}X^m E}, \sigma_{\mathcal{T}X^m E} \big) = \Bigg( \sum_T P_\mathcal{T}(T) \delta\big( \ketbra{\phi_{X^m E}}, \ketbra{\tilde{\phi}_{X^m E}^T} \big) \Bigg)^2 \leq\sum_T P_\mathcal{T}(T) |\epsilon_T^\perp|^2, 
		\end{align}
		where we used Jensen's inequality and properties of the trace norm. The last term is the probability that, when choosing $T$ according to $P_{\mathcal{T}}$ and measuring the state of $X^m$ in bases $\boldsymbol{x}$ we get an outcome  $\boldsymbol{r'} \notin \mathcal{B}_{T}$. We write  
		\begin{align}
		\sum_T P_\mathcal{T}(T)|\epsilon_T^\perp|^2 = \text{Pr}_{\mathcal{T}}[\boldsymbol{r'} \notin \mathcal{B}_{T}] = \text{Pr}_{\mathcal{T}}[r_H(\boldsymbol{r'}_{|\bar{T}}, \boldsymbol{r}_{|\bar{T}}) - r_H(\boldsymbol{r'}_{|T}, \boldsymbol{r}_{|T}) > \zeta].
		\end{align}
		Then, we can use Corollary 4 from \cite{LPTRG13}, which states that the above probability is negligible in $n$ and gives us an upper bound for $\delta\big( \rho_{\mathcal{T}X^m E}, \sigma_{\mathcal{T}X^m E} \big)$. In particular, given the set $[m]$ with $(1+t)n$ elements, we apply the aforementioned corollary for a random subset $\mathcal{T}$ of size $tn$ and its complement $\bar{\mathcal{T}}$ of size $n$. Denoting by $\mu_{\mathcal{T}}$ and $\mu_{\bar{\mathcal{T}}}$, respectively,  the averages of these subsets, we obtain
		\begin{align}
		\text{Pr}[\mu_{\bar{\mathcal{T}}} - \mu_{\mathcal{T}} \geq \zeta] \leq \exp( -\frac{2 \zeta^2t^2n^2}{(nt+1)(t+1)}).
		\end{align}
		Hence, we have:
		\begin{align}
		\delta\big( \rho_{\mathcal{T}X^m E}, \sigma_{\mathcal{T}X^m E} \big) \leq  \sum_T P_\mathcal{T}(T)|\epsilon_T^\perp|^2 \leq \exp( -\frac{2 \zeta^2t^2n^2}{(nt+1)(t+1)})=: \epsilon,
		\end{align}
		concluding the proof of Lemma \ref{lemma:dishonest_bob_2}, i.e. that the real state \eqref{eq:real_definition} generated by the protocol until the \textit{Computation Phase}  is  $\epsilon-$close to the ideal state \eqref{eq:ideal_definition}.
	\end{proof}
	
	It is now straightforward to complete the proof of property 1. of Lemma \ref{lemma:wrole_dishonest_bob}. 
	
	To obtain the states $\sigma_{\mathbf{F} B'}$ and $\rho_{\mathbf{F} B'}$ from the states $\sigma_{\mathcal{T}X^m E}$ and $\rho_{\mathcal{T}X^m E}$, respectively, we need to apply the operator $V^{\bm{b}}_{\bm{a}}\Tr_{X^m_{|T}\,E_{|T}}[\, \cdot\,] V^{\bm{b} \dagger}_{\bm{a}}$. This operator  is a CPTP map, being the composition of the two CPTP maps, $V^{\bm{b}}_{\bm{a}}$ and $\Tr$. Since the trace distance between two density matrices does not increase under CPTP maps (see Lemma 7 in \cite{U17}), the final states indeed satisfy  property 1. of Lemma \ref{lemma:wrole_dishonest_bob}, namely 
	\begin{align}
	\sigma_{\mathbf{F} B'} \approx_{\epsilon} \rho_{\mathbf{F} B'}.
	\end{align}

	For the rest of this proof dishonest Bob's system $B'$ is identified with $\mathbf{Y}\, E$, where $\mathbf{Y}$ corresponds to the classical information leaking to him through the output of the WROLE and $E$ is, in general, a quantum auxiliary system that he might also hold. Consequently, from now on we write $\sigma_{\mathbf{F} B'}$ as $\sigma_{\mathbf{F} \mathbf{Y}E}$. Now, we have to prove property 2. of Lemma~\ref{lemma:wrole_dishonest_bob},  i.e., obtain the corresponding lower bound on min-entropy with respect to $\sigma_{\mathbf{F} \mathbf{Y}E}$. 
	We start with  the following Lemma~\ref{XnBound}:
	
	\begin{lemma}[Corollary 4.4 in \cite{DFLSS09}]
		\label{XnBound}
		Let $err := r_H(\boldsymbol{r'}_{|T}, \boldsymbol{r}_{|T}) \leq 1-\frac{1}{d}$ be the error measured by Alice  while measuring the state of $X^m_{|T}$ according to her choice of $T$, and let $\sigma_{\mathbf{X} E} := \ketbra{\psi}_{\mathbf{X} E}$ be the state to which the ideal state $\sigma_{\mathcal{T}X^m E}$ collapses after this measurement. Following \eqref{eq:ideal_definition} and \eqref{eq:ideal_definition_2}, we  write  $\ket{\psi}_{\mathbf{X} E} = \sum_{\boldsymbol{z}\in \mathcal{B}} \alpha_{\boldsymbol{z}}\ket{e^{\boldsymbol{x}}_{\boldsymbol{z}}} \ket{\psi^{\boldsymbol{z}}_{E}}$  for some $\ket{\psi^{\boldsymbol{z}}_{E}}$ and  $\mathcal{B} = \{ \boldsymbol{z}\in \mathbb{Z}_d^{n}: r_H(\boldsymbol{z}, \boldsymbol{r}_{|\bar{T}})\leq err + \zeta \}$ with $\zeta>0$. Then, we have: 
		\begin{align}
		H_{\min}(\mathbf{X}|E)_{\sigma_{\mathbf{X} E}} \geq -h_d(err + \zeta) n\log d,
		\label{eq:min_entropy_1_stage}
		\end{align}
		where $h_d(x)$ is given in Definition \ref{def:q-ary}.
		
	\end{lemma}
	\begin{proof}
				We start by defining the state $\tilde{\sigma}_{\mathbf{X} E}: = \sum_{\bm{z}\in \mathcal{B}} |\alpha_{\bm{z}}|^2 \ketbra{e^{\boldsymbol{x}}_{\boldsymbol{z}}} \otimes \ketbra{\psi^{\boldsymbol{z}}_{E}}$. 
		Then, by applying Lemma 3.1.13 from \cite{R06}, we obtain
		\begin{align}
		H_{\text{min}}(\mathbf{X} |E)_{\sigma_{\mathbf{X} E}} \geq H_{\text{min}}(\mathbf{X} | E)_{\tilde{\sigma}_{\mathbf{X} E}} - \log |\mathcal{B}|.
		\end{align}

		Since $\tilde{\sigma}_{\mathbf{X} E}$  is a classical-quantum state, its min-entropy cannot be negative, that is $H_{\text{min}}(\mathbf{X} | E)_{\tilde{\sigma}_{\mathbf{X} E}} \geq 0$, thus $H_{\text{min}}(\mathbf{X} |E)_{\sigma_{\mathbf{X} E}} \geq -\log |\mathcal{B}|$. 
		The lower bound shown in \eqref{eq:min_entropy_1_stage} follows directly from Definition \ref{def:q-ary} considering the Hamming ball around $\boldsymbol{r}_{|\bar{T}}$ with radius $n(err + \zeta)$.
	\end{proof}
	
	To complete the proof of property 2. of Lemma \ref{lemma:wrole_dishonest_bob} and find a lower bound on $H_\text{min}(\mathbf{F}|\mathbf{Y}\, E)\sigma_{\mathbf{F} \mathbf{Y}E}$  we need to relate it with $H_{\text{min}}(\mathbf{X} |E)_{\sigma_{\mathbf{X} E}}$, for which we just derived a lower bound. In what follows, we 
	adapt the notation from \cite{Dupuis2015}:
	
	\begin{itemize}
		\item $\Phi_{\mathbf{X} \bar{\mathbf{X}}} = \ketbra{\Phi}_{\mathbf{X} \bar{\mathbf{X}}}$, where $\ket{\Phi}_{\mathbf{X} \bar{\mathbf{X}}} = \sum_{\bm{s}} \ket{e^{\bm{x}}_{\bm{s}}}_{\mathbf{X} } \otimes \ket{e^{\bm{x}}_{\bm{s}}}_{\bar{\mathbf{X}}}$ for all basis choices   $\bm{x}  \in\mathbb{Z}^n_d$, and
		\item $\Phi_{(\bm{a},\bm{b})} = \ketbra{\Phi_{(\bm{a},\bm{b})}} = \sum_{\bm{s s'}} V^{\bm{b}}_{\bm{a}} \ketbra{e^{\bm{x}}_{\bm{s}}}{e^{\bm{x}}_{\bm{s'}}} V^{\bm{b}\dagger}_{\bm{a}} \otimes \ketbra{e^{\bm{x}}_{\bm{s}}}{e^{\bm{x}}_{\bm{s'}}}$, with $\ket{\Phi_{(\bm{a},\bm{b})}} = (V^{\bm{b}}_{\bm{a}} \otimes \mathds{1}) \ket{\Phi}_{\mathbf{X} \bar{\mathbf{X}}}$,  for $(\bm{a},\bm{b})\in\mathbb{Z}^2_d$
	\end{itemize}
	and we use the following properties of the operators $V_a^b$, which can be derived from \eqref{eq:main_relation}:
	
	\begin{enumerate}[(a)]
		\item $V^b_a \ketbra{e^x_i} V^{b \dagger}_a = \ketbra{e^x_{ax - b + i}}$, \label{prop:a}
		\item $V^b_a = \sum_j c_{a,b,j,x} \ketbra{e^x_{ax - b + j}}{e^x_j}$,\label{prop:b}
		\item $V^{b\dagger}_a = \sum_j c_{a,b,j,x}^* \ketbra{e^x_j}{e^x_{ax - b + j}}$.\label{prop:c}
		\end{enumerate}
	The following lemma,  which is an adaptation of Theorem 12 in \cite{Dupuis2015} to our case, provides a lower bound for $H_\text{min}(\mathbf{F}|\mathbf{Y}\, E)_{\sigma_{\mathbf{F}\mathbf{Y} E}}$ in terms of $H_{\text{min}}(\mathbf{X} |E)_{\sigma_{\mathbf{X}E}}$.
	\begin{lemma} 
		\label{lem:GnBound}
		Let $\mathbf{X}$ denote our $n$-qudit system and $\sigma_{\mathbf{X} E}$ be the ideal quantum state to which the system collapsed after Alice's test measurements, as introduced before. Then, we have:
		\begin{align}
		\text{H}_\text{min}(\mathbf{F}|\mathbf{Y}\, E)_{\sigma_{\mathbf{F}\mathbf{Y} E}} \geq \frac{1}{2} \big(n\log d + \text{H}_\text{min}(\mathbf{X}| E)_{\sigma_{\mathbf{X} E}}\big),
		\end{align}
		where 
		\begin{align}
		\sigma_{\mathbf{F}\mathbf{Y} E} = \frac{1}{d^{2n}} \sum_{(\bm{a}, \bm{b}) \in \mathbb{Z}^{2n}_d} \ketbra{e^{\bm{x}}_{\bm{a}},e^{\bm{x}}_{\bm{b}}}\otimes V^{\bm{b}}_{\bm{a}} \, \sigma_{\mathbf{X} E}\,  V^{\bm{b}\dagger}_{\bm{a}},
		\label{eq:VappliedX}
		\end{align}
		is the state obtained when $V^{\bm{b}}_{\bm{a}}$ is applied to the system $\mathbf{X}$ according to $\mathbf{F}$. 
	\end{lemma}
	
	\begin{proof}
		This proof is an adaptation of the proof of Theorem 12 in \cite{Dupuis2015} to our case.
		Let us fix $x\in\mathbb{Z}_d$ and write $\ket{i} = \ket{e^x_i}$  for short. 
		$V^b_a$ is a CPTP map, and it is known that the min-entropy does not decrease whenever a CPTP map is applied. However, this is not enough  to prove the security of the protocol and determine a lower bound. We need a more refined expression relating  $H_{\min}(\mathcal{M}(\mathbf{X})|E)$ and  $H_{\min}(\mathbf{X}|E)$ for some CPTP map $\mathcal{M}$. This is given by
		the following Lemma~\ref{thm:entaglementSamplingResult},  which is an adaptation of Theorem 1 in \cite{Dupuis2015} to our case:
		
		\begin{lemma}[Theorem 1, \cite{Dupuis2015}] 
			\label{thm:entaglementSamplingResult}
			Let $\mathbf{X}$ denote our system of  $n$ qudits, and  $\mathcal{M}_{\mathbf{X}\rightarrow \mathbf{F}\mathbf{Y}}$ be a CP map such that $((\mathcal{M}^\dagger \circ \mathcal{M})_{\mathbf{X}}\otimes \text{id}_{\bar{\mathbf{X}}})(\Phi_{\mathbf{X}\bar{\mathbf{X}}}) = \sum_{(\bm{a},\bm{b})\in\mathbb{Z}^{2n}_d} \lambda_{(\bm{a},\bm{b})} \Phi_{(\bm{a},\bm{b})}$. As above, let $\sigma_{\mathbf{X} E}$ be the ideal quantum state to which the system collapsed after Alice's test measurements.  Then, for any partition of $\mathbb{Z}^{2n}_d = \mathfrak{S}_+ \cup \mathfrak{S}_-$ into subsets $\mathfrak{S}_+$ and $\mathfrak{S}_-$, we have 
			\begin{align}
			2^{-\text{H}_2(\mathbf{F}\mathbf{Y} | E)_{\sigma_{\mathbf{F}\mathbf{Y}E} | \sigma_{\mathbf{X}E}}} \leq \sum_{(\bm{a},\bm{b})\in\mathfrak{S}_+} \lambda_{(\bm{a},\bm{b})} 2^{-\text{H}_2(\mathbf{X} | E)_{\sigma_{\mathbf{X}E}}} + (\max_{(\bm{a},\bm{b})\in\mathfrak{S}_-} \lambda_{(\bm{a},\bm{b})}) d^n,
			\end{align}
			where, in general, for a (not necessarily normalized) quantum state $\rho_{AB}\in \mathcal{P}(\mathcal{H}_A\otimes\mathcal{H}_B)$, $\text{H}_2(A|B)$   is the so-called \textit{collision entropy}~\cite{R06}, given as 
			\begin{align} 
			\text{H}_2(A|B)_{\rho_{AB}}=-\log \left(\Tr{\left(\rho_{B}^{-1/4}\rho_{AB}\rho_B^{-1/4}\right)^2}\right).
			\end{align}
			If we further condition on a general quantum state $\sigma_B\in\mathcal{P}(\mathcal{H}_B)$, we have 
			\begin{align}
			\text{H}_2(A|B)_{\rho_{AB}|\sigma_B}=-\log \left(\Tr{\left(\sigma_{B}^{-1/4}\rho_{AB}\sigma_B^{-1/4}\right)^2}\right).
			\end{align}
		\end{lemma}
		We can apply the above Lemma~\ref{thm:entaglementSamplingResult}  with $\mathcal{M}_{\mathbf{X}\rightarrow \mathbf{F}\mathbf{Y}} = \mathcal{N}^{\otimes n}$, where 
		\begin{align}
		\mathcal{N}(\sigma_X) = \frac{1}{d^2}\sum_{(a,b)\in\mathbb{Z}^2_d}\mathcal{N}_{a,b} \sigma_X \mathcal{N}_{a,b}^\dagger = \frac{1}{d^2} \sum_{(a,b)\in\mathbb{Z}^2_d} \bigg( \ket{a,b} \otimes V^b_a \bigg) \sigma_X \bigg( \bra{a,b} \otimes V^{b\dagger}_a \bigg).
		\end{align}
		 The operator $\mathcal{N}$ applies the operator $V^b_a$ to the single system $X$ and saves its choice $(a,b)$ to a new record in the  $F$ space.  
		To continue, we want to write $((\mathcal{M}^\dagger \circ \mathcal{M} )_{\mathbf{X}} \otimes \text{id}_{\bar{\mathbf{X}}}) (\Phi_{\mathbf{X} \bar{\mathbf{X}}})$  as a linear combination of $\Phi_{(\bm{a},\bm{b})}$, i.e.
	\begin{align}
	((\mathcal{M}^\dagger \circ \mathcal{M} )_{\mathbf{X}} \otimes \text{id}_{\bar{\mathbf{X}}}) (\Phi_{\mathbf{X} \bar{\mathbf{X}}}) = \sum_{(\bm{a},\bm{b})\in \mathbb{Z}^{2n}} \lambda_{(\bm{a},\bm{b})} \Phi_{(\bm{a},\bm{b})}
	\end{align}
		First, note that 
		\begin{align}
		\mathcal{N}(\ketbra{i}{j}) & = \frac{1}{d^2} \sum_{(a,b)\in\mathbb{Z}^2_d} \bigg( \ket{a,b} \otimes V^b_a \bigg) \ketbra{i}{j} \bigg( \bra{a,b} \otimes V^{b\dagger}_a \bigg) \nonumber\\
		&\overset{\ref{prop:a}}{=} \frac{1}{d^2} \sum_{(a,b)\in\mathbb{Z}^2_d} \ketbra{a,b} \otimes \ketbra{ax-b+i}{ax-b+j}, \label{eqn:Nresult}
		\end{align}
		where we used the property \ref{prop:a}.  We proceed to compute $\mathcal{N}^\dagger \circ \mathcal{N} \ketbra{i}{j}$:
		\begin{align}
		\mathcal{N}^\dagger \circ \mathcal{N} \ketbra{i}{j} &= \frac{1}{d^2} \sum_{(a',b')\in\mathbb{Z}^2_d} \mathcal{N}^\dagger_{a' b'} (\mathcal{N}\ketbra{i}{j})) \mathcal{N}_{a' b'} \nonumber\\
		&\overset{\eqref{eqn:Nresult}}{=} \frac{1}{d^4} \sum_{(a',b'),(a,b)\in\mathbb{Z}^2_d} \Big( \bra{a',b'} \otimes V^{b'\dagger}_{a'} \Big) \ketbra{a,b}\otimes \ketbra{ax-b+i}{ax-b+j} \Big( \ket{a',b'} \otimes V^{b'}_{a'}  \Big) \nonumber\\
	&\overset{\ref{prop:b}, \ref{prop:c}}{=} \frac{1}{d^4}\sum_{(a,b)\in\mathbb{Z}^2_d} V^{b\dagger}_{a} \ketbra{ax-b+i}{ax-b+j} V^b_a = \frac{1}{d^4} \sum_{(a,b)\in\mathbb{Z}^2_d} \ketbra{i}{j} = \frac{1}{d^2} \ketbra{i}{j}, \label{eqn:NdaggerNresult} 
		\end{align}
		and		
		\begin{align}
		((\mathcal{N}^\dagger \circ \mathcal{N})_{X} \otimes \text{id}_{\bar{X}}) (\Phi_{X \bar{X}})& = \sum_{i,j} \mathcal{N}^\dagger \circ \mathcal{N} (\ketbra{i}{j}) \otimes \ketbra{i}{j} \nonumber\\ &
		\overset{\text{eq.}(\ref{eqn:NdaggerNresult}) }{=} \sum_{i,j} \Big(  \frac{1}{d^2} \ketbra{i}{j} \Big) \otimes \ketbra{i}{j}\nonumber\\ & =\frac{1}{d^2} \sum_{i,j} \ketbra{i}{j} \otimes \ketbra{i}{j}= \frac{1}{d^2} \Phi_{X \bar{X}}.
		\label{eqn:NdaggerNidresult}
		\end{align}
		Therefore, we have 
		\begin{align}
		\big((\mathcal{N}^\dagger \circ \mathcal{N})_{X} \otimes \text{id}_{\bar{X}}\big) (\Phi_{X \bar{X}}) = \frac{1}{d^2} \Phi_{(0,0)},
		\end{align}
		from which we easily see  
		\begin{align}
		\big((\mathcal{M}^\dagger \circ \mathcal{M} )_{\mathbf{X}} \otimes \text{id}_{\bar{\mathbf{X}}}\big) (\Phi_{\mathbf{X} \bar{\mathbf{X}}}) = \frac{1}{d^{2n}} \Phi_{(\bm{0},\bm{0})}.
		\end{align}
		Consequently, 
		\begin{align}
		\lambda_{(\bm{a}, \bm{b})}= \left\{
		\begin{array}{ll}
		\frac{1}{d^{2n}} & \text{if } (\bm{a}, \bm{b}) = (\bm{0},\bm{0})\\
		0 & \text{otherwise.} \\
		\end{array} 
		\right. 
		\end{align}
		Now, we want to choose the partition that gives us the best lower bound on the collision entropy, i.e. decreases the r.h.s of the following relation:
		\begin{align}
			2^{-\text{H}_2(\mathbf{F}\mathbf{Y} | E)_{\sigma_{\mathbf{F}\mathbf{Y} E}}} \leq \sum_{(\bm{a},\bm{b}) \in \mathfrak{S}_+} \lambda_{(\bm{a},\bm{b})} 2^{-\text{H}_2(\mathbf{X}| E)_{\sigma_{\mathbf{X} E}}} + (\max_{(\bm{a},\bm{b}) \in \mathfrak{S}_-} \lambda_{(\bm{a},\bm{b})}) d^n.
		\end{align}
		Note that  we dropped the conditioning on the state $\sigma_{\mathbf{X}E}$ at $\text{H}_2(\mathbf{F}\mathbf{Y} | E)_{\sigma_{\mathbf{F}\mathbf{Y} E}}$. This is because the map $\mathcal{M}$ is trace-preserving (for a detailed explanation see \cite{Dupuis2015} below Theorem 1). For our case there are just two types of partitions: the case where $0 \in \mathfrak{S}_+$ and the case where $0 \in \mathfrak{S}_-$. If $0 \in \mathfrak{S}_+$:
			\begin{align}
			\text{r.h.s} = \sum_{(\bm{a}, \bm{b}) \in \mathfrak{S}_+} \lambda_{s} 2^{-\text{H}_2(\mathbf{X}| E)_{\sigma_{\mathbf{X} E}}} = \frac{1}{d^{2n}} 2^{-\text{H}_2(\mathbf{X}| E)_{\sigma_{\mathbf{X} E}}} \leq \frac{1}{d^n}.
		\end{align}
		The last inequality holds because $-n\log d \leq \text{H}_2(\mathbf{X}| E)_{\sigma_{\mathbf{X} E}} \leq n\log d$. In fact,
		\begin{align}
			\frac{1}{d^{2n}} 2^{-\text{H}_2(\mathbf{X}| E)_{\sigma_{\mathbf{X} E}}} \leq \frac{1}{d^n}
			\iff 2^{-\text{H}_2(\mathbf{X}| E)_{\sigma_{\mathbf{X} E}}} \leq d^n 
			\iff \text{H}_2(\mathbf{X}| E)_{\sigma_{\mathbf{X} E}}\geq -n\log d.
		\end{align}
	
		If $0 \in \mathfrak{S}_-$, $\text{r.h.s} = \frac{1}{d^n}$ which, as we have just seen by the previous inequality, does not provide a better lower bound on the collision entropy whenever $\text{H}_2(\mathbf{X}| E)_{\sigma_{\mathbf{X} E}} \neq -n\log d$.
		So, choosing any partition such that $0 \in \mathfrak{S}_+$, we get
		\begin{align}
			2^{-\text{H}_2(\mathbf{F}\mathbf{Y} | E)_{\sigma_{\mathbf{F}\mathbf{Y} E}}} & \leq \sum_{(\bm{a}, \bm{b}) \in \mathfrak{S}_+} \lambda_{(\bm{a}, \bm{b})} 2^{-\text{H}_2(\mathbf{X}| E)_{\sigma_{\mathbf{X} E}}} + (\max_{(\bm{a}, \bm{b}) \in \mathfrak{S}_-} \lambda_{(\bm{a}, \bm{b})}) d^n\nonumber\\ &=\frac{1}{d^{2n}} 2^{-\text{H}_2(\mathbf{X} | E)_{\sigma_{\mathbf{X} E}}},
		\end{align}
		from which we conclude
		\begin{align}
		\text{H}_2(\mathbf{F}\mathbf{Y} | E)_{\sigma_{\mathbf{F}\mathbf{Y} E}} \geq 2n\log d + \text{H}_2(\mathbf{X}|E)_{\sigma_{\mathbf{X} E}}.  \label{eqn:firstineq} 
		\end{align}
		
		In order to relate $\text{H}_2(\mathbf{F}\mathbf{Y} | E)_{\sigma_{\mathbf{F}\mathbf{Y} E}}$ with $\text{H}_2(\mathbf{F} | \mathbf{Y} E)_{\sigma_{\mathbf{F}\mathbf{Y} E}}$, we use the chain rule proved in Proposition 8 of \cite{MDSFT13}:
		\begin{align}
		\text{H}_2(\mathbf{F} | \mathbf{Y} E)_{\sigma_{\mathbf{F}\mathbf{Y} E}} \geq \text{H}_2(\mathbf{F}\mathbf{Y} | E)_{\sigma_{\mathbf{F}\mathbf{Y} E}} - \log \text{rank}(\sigma_{\mathbf{Y}}) \geq \text{H}_2(\mathbf{F}\mathbf{Y} | E)_{\sigma_{\mathbf{F}\mathbf{Y} E}} - n \log d,
		\label{eqn:chainRule} 
		\end{align}
		since $\log \text{rank}(\sigma_{\mathbf{Y}}) \leq n \log d$. Combining (\ref{eqn:firstineq}) and (\ref{eqn:chainRule}), we get the desired result:
		\begin{align}
		\text{H}_2(\mathbf{F}|\mathbf{Y}\, E)_{\sigma_{\mathbf{F}\mathbf{Y} E}} \geq n\log d + \text{H}_2(\mathbf{X}| E)_{\sigma_{\mathbf{X} E}}.
		\end{align}
		To express the above relation in terms of the min-entropy instead of the collision entropy, we start by noticing that 
		the state given in \eqref{eq:VappliedX} can be written as 
		\begin{align}
		\sigma_{\mathbf{F}\mathbf{Y} E} = \frac{1}{d^{2n}} \sum_{(\bm{a}, \bm{b}) \in \mathbb{Z}^{2n}_d} \ketbra{\bm{a},\bm{b}}\otimes V^{\bm{b}}_{\bm{a}} \, \sigma_{\mathbf{X} E}\,  V^{\bm{b}\dagger}_{\bm{a}} = \frac{1}{d^{2n}} \sum_{(\bm{a}, \bm{b}) \in \mathbb{Z}^{2n}_d} \ketbra{\bm{a},\bm{b}}\otimes \sigma_{\mathbf{X} E}^{\bm{a},\bm{b}},
		\end{align}
		which is a classical-quantum state. Therefore, we can use Lemma 18 from \cite{Dupuis2015} to obtain
		\begin{align}
		\text{H}_{\text{min}}(\mathbf{F}|\mathbf{Y}\, E)_{\sigma_{\mathbf{F}\mathbf{Y} E}}  \geq \frac{1}{2}\left(n\log d + \text{H}_2(\mathbf{X}| E)_{\sigma_{\mathbf{X} E}}\right).
		\end{align}
		Furthermore,  $\sigma_{\mathbf{X}E}$  is a general quantum state, and from Lemma 17 in \cite{Dupuis2015} we have 
		\begin{align}
		\text{H}_{\text{min}}(\mathbf{F}|\mathbf{Y}\, E)_{\sigma_{\mathbf{F}\mathbf{Y} E}}  \geq \frac{1}{2}\left(n\log d + \text{H}_{\text{min}}(\mathbf{X}| E)_{\sigma_{\mathbf{X} E}} \right).
		\end{align}
	\end{proof}
	To complete the proof of property 2. of Lemma \ref{lemma:wrole_dishonest_bob}, we  combine Lemma~\ref{XnBound} and Lemma~\ref{lem:GnBound},  and obtain:
	\begin{align}
	H_{\min}(\mathbf{F} | \mathbf{Y} E)_{\sigma_{\mathbf{F}\mathbf{Y} E}} \geq \frac{1}{2}(n\log d - h_d(\zeta)n\log d) = \frac{n\log d}{2}(1 - h_d(\zeta)),
	\end{align}
	for $err = 0$.
\end{proof}
\section{Lemma B.1}

\begin{lemma}
	Let $\rho_{XB} \in \mathcal{P}(\mathcal{H}_{X}\otimes \mathcal{H}_B)$ be a cq-state and let $f:\mathcal{X} \rightarrow \mathcal{X} $ be a fixed bijective function. Then,
\begin{align}
H_{\min}(X|E)_{\rho} \leq H_{\min}(f(X)|B)_{\rho}.
\end{align}
	\label{lemma:bijectivefunction}
\end{lemma}
\begin{proof}
	Consider the unitary operator,
$
U = \sum_x \ketbra{f(x)}{x}
$. Checking that $U$ is indeed unitary:
	\begin{align}
	U U^{\dagger} = \left(\sum_x \ketbra{f(x)}{x}\right)\left(\sum_{x'} \ketbra{x'}{f(x')}\right) 
	= \sum_x \ketbra{f(x)} = I,
	\end{align}
	where in the last step we used the fact that the function $f$ is a bijection. The same holds for $U^{\dagger}U = I$. Now, observe that
	\begin{align}
		H_{\min}(f(X)|B) &= -\log \max_{\{M_x\}_x} \sum_x p_x \tr\left[M_x \rho_{f(x)}^B\right]\nonumber\\
		&= -\log \max_{\{M_x\}_x} \sum_x p_x \tr\left[M_x U \rho_{x}^B U^{\dagger}\right]\nonumber\\
		&= -\log \max_{\{M_x\}_x} \sum_x p_x \tr\left[U^{\dagger} M_x U \rho_{x}^B \right]. 
	\end{align}
	
	Note that $\left\{ N_x \right\}_x= \left\{U^{\dagger} M_x U\right\}_x$ is also a POVM: they are all positive semidefinite operators and it sums up to unity. Therefore, we have that $\left\{U^{\dagger} M_x U\right\}_x$ can only decrease the space of possible POVMs, i.e
	\begin{align}
	\max_{\{M_x\}_x} \sum_x p_x \tr\left[U^{\dagger} M_x U \rho_{x}^B \right] \leq \max_{\{M_x\}_x} \sum_x p_x \tr\left[ M_x \rho_{x}^B \right].
	\end{align}
This means 
	\begin{align}
	H_{\min}(f(X)|B) \geq -\log \max_{\{M_x\}_x} \sum_x p_x \tr\left[ M_x \rho_{x}^B \right] = H_{\min}(X|B).
	\end{align}
\end{proof}

\section{Derivation of  \eqref{eq:general_relation}}\label{app:generalrelation}

The relation \eqref{eq:general_relation} can be easily deduced by considering the following property  from \cite{DEBZ10} (Equation $(2.56)$ in section 2.4.2):
\begin{align}
\alpha^i_k \alpha^i_l = \alpha^i_{k\oplus l} \omega^{i\odot k\odot l}.
\end{align}
We have
\begin{align}
	V^b_a\ket{e^i_r} &= \frac{1}{\sqrt{N}} \sum_{k,l = 0}^{N-1} \ket{k\oplus a}\omega^{(k\oplus a)\odot b}\omega^{\ominus r\odot l}\braket{k}{l} \alpha^{i*}_{\ominus l} \\
	&=\frac{1}{\sqrt{N}} \sum_{l=0}^{N-1}\ket{l}\omega^{l\odot b}\omega^{\ominus r\odot(l\ominus a)} \alpha^{i*}_{\ominus(l\ominus a)} \\
	&= \frac{1}{\sqrt{N}} \sum_{l=0}^{N-1}\ket{l}\omega^{l\odot b \ominus r\odot( l\ominus a)} \big(\omega^{\ominus(i\odot a\odot(\ominus l))} \alpha^i_a \alpha^i_{\ominus l} \big)^*\\
	&= \omega^{r\odot a}\alpha^{i*}_a \frac{1}{\sqrt{N}} \sum_{l=0}^{N-1}\ket{l}\omega^{l\odot b \ominus r\odot l} \omega^{\ominus(i\odot a\odot l)}  \alpha^{i*}_{\ominus l}\\
	&= \omega^{r\odot a}\alpha^{i*}_a \frac{1}{\sqrt{N}} \sum_{l=0}^{N-1}\ket{l}\omega^{\ominus(i\odot a\ominus b \oplus r)\odot l}  \alpha^{i*}_{\ominus l}\\
	&=\omega^{r\odot a} \alpha^{i*}_a \ket{e^i_{i\odot a\ominus b\oplus r}}.
\end{align}

\end{document}